\documentclass{article}
\usepackage{amsmath,amssymb,bbm,amsthm}
\usepackage{fullpage}
\usepackage{color,xcolor,hyperref}
\usepackage{cleveref,algorithm,xspace}
\usepackage{thm-restate}

\newcommand{\f}{\frac}
\newcommand{\cd}{\cdot}
\newcommand{\bn}{\binom}

\newcommand{\lds}{\ldots}

\newcommand{\bs}{\backslash}
\newcommand{\s}{\subseteq}

\newcommand{\BE}{\begin{enumerate}}
\newcommand{\EE}{\end{enumerate}}
\newcommand{\im}{\item}
\newcommand{\BI}{\begin{itemize}}
\newcommand{\EI}{\end{itemize}}

\newcommand{\eps}{\epsilon}
\newcommand{\e}{\epsilon}

\newcommand{\be}{\beta}
\newcommand{\om}{\omega}
\newcommand{\Om}{\Omega}

\newcommand{\m}{\mathcal}

\newcommand{\E}{\mathbb E}

\newcommand{\poly}{\textup{poly}}

\newcommand{\lp}{\left(}
\newcommand{\rp}{\right)}
\newcommand{\lb}{\left[}
\newcommand{\rb}{\right]}
\newcommand{\lmt}{\left[\begin{matrix}}
\newcommand{\rmt}{\end{matrix}\right]}

\newtheorem{theorem}{Theorem}
\newtheorem{lemma}{Lemma}
\newtheorem{definition}{Definition}
\newtheorem{corollary}{Corollary}
\newtheorem{observation}{Observation}
\newtheorem{claim}{Claim}
\newtheorem{subclaim}{Subclaim}
\newtheorem{assumption}{Assumption}
\newtheorem{reduction}{Reduction}
\newcommand{\BT}{\begin{theorem}}
\newcommand{\ET}{\end{theorem}}
\newcommand{\BL}{\begin{lemma}}
\newcommand{\EL}{\end{lemma}}
\newcommand{\BD}{\begin{definition}}
\newcommand{\ED}{\end{definition}}
\newcommand{\BC}{\begin{corollary}}
\newcommand{\EC}{\end{corollary}}
\newcommand{\BO}{\begin{observation}}
\newcommand{\EO}{\end{observation}}
\newcommand{\BCL}{\begin{claim}}
\newcommand{\ECL}{\end{claim}}
\newcommand{\BP}{\begin{proof}}
\newcommand{\EP}{\end{proof}}
\newcommand{\BPS}{\begin{proof}[Proof (Sketch)]}
\newcommand{\EPS}{\end{proof}}
\newcommand{\BA}{\begin{assumption}}
\newcommand{\EA}{\end{assumption}}
\newcommand{\BR}{\begin{reduction}}
\newcommand{\ER}{\end{reduction}}
\newcommand{\BSCL}{\begin{subclaim}}
\newcommand{\ESCL}{\end{subclaim}}

\Crefname{observation}{Observation}{Observations}
\Crefname{assumption}{Assumption}{Assumptions}
\Crefname{reduction}{Reduction}{Reductions}
\Crefname{claim}{Claim}{Claims}
\Crefname{subclaim}{Subclaim}{Sublaims}
\newcommand{\Os}{{O^\star}}
\renewcommand{\bar}{\overline}
\newcommand{\tw}{\textup{tw}}

\newcommand{\Inf}{\texttt{\textup{Infeasible}}\xspace}
\newcommand{\cnt}{\texttt{\textup{count}}\xspace}
\newcommand{\cntA}{\texttt{\textup{countA}}\xspace}
\newcommand{\cntB}{\texttt{\textup{countB}}\xspace}
\newcommand{\cntthree}{\texttt{\textup{count3}}\xspace}
\newcommand{\cntzero}{\texttt{\textup{count0}}\xspace}

\newcommand{\Heads}{\texttt{\textup{Heads}}\xspace}

\newcommand{\treewidthDP}{\textup{\texttt{treewidthDP}}\xspace}
\newcommand{\forestDP}{\textup{\texttt{forestDP}}\xspace}

\newenvironment{subproof}[1][\proofname]{%
  \begin{proof}[#1]%
}{%
  \end{proof}%
}

\usepackage{algorithm}
\usepackage[noend]{algpseudocode}
\algnewcommand{\IIf}[1]{\State\algorithmicif\ #1\ \algorithmicthen}
\algnewcommand{\EndIIf}{\unskip\ \algorithmicend\ \algorithmicif}
\algrenewcommand\algorithmiccomment[2][\normalsize]{{#1\hfill\(\triangleright\) \emph{#2}}}
\newcommand{\LineIf}[2]{ \State \algorithmicif\ {#1}\ \algorithmicthen\ {#2} }

\makeatletter
\newcounter{algocounter}
\@ifpackageloaded{hyperref}%
  {\newcommand{\mylabel}[2]
    {\refstepcounter{algocounter}\protected@write\@auxout{}{\string\newlabel{#1}{{\textcolor{black}{\textup{#2}}}{\thepage}%
      {\@currentlabelname}{\@currentHref}{}}}}}%
\makeatother

\begin{document}

\title{Detecting Feedback Vertex Sets of Size $k$ in $\Os(2.7^k)$ Time}
\author{Jason Li\thanks{Carnegie Mellon University, \url{jmli@andrew.cmu.edu}.} \and Jesper Nederlof\thanks{Eindhoven University of Technology, \url{j.nederlof@tue.nl}. Supported by the Netherlands Organization for Scientific Research (NWO) under project no. 639.021.438 and 024.002.003. and the European Research Council under project no. 617951.}}
\date{\today}
\maketitle
\begin{abstract}
        In the Feedback Vertex Set problem, one is given an undirected graph $G$ and an integer $k$, and one needs to determine whether there exists a set of $k$ vertices that intersects all cycles of $G$ (a so-called feedback vertex set).
        Feedback Vertex Set is one of the most central problems in parameterized complexity: It served as an excellent test bed for many important algorithmic techniques in the field such as Iterative Compression~[Guo et al. (JCSS'06)], Randomized Branching~[Becker et al. (J. Artif. Intell. Res'00)] and Cut\&Count~[Cygan et al. (FOCS'11)].
        In particular, there has been a long race for the smallest dependence $f(k)$ in run times of the type $O^\star(f(k))$, where the $O^\star$ notation omits factors polynomial in $n$.
        This race seemed to be run in 2011, when a randomized $O^\star(3^k)$ time algorithm based on Cut\&Count was introduced. 
        
        In this work, we show the contrary and give a $O^\star(2.7^k)$ time randomized algorithm.
        Our algorithm combines all mentioned techniques with substantial new ideas:
        First, we show that, given a feedback vertex set of size $k$ of bounded average degree, a tree decomposition of width $(1-\Omega(1))k$ can be found in polynomial time.
        Second, we give a randomized branching strategy inspired by the one from~[Becker et al. (J. Artif. Intell. Res'00)] to reduce to the aforementioned bounded average degree setting. 
        Third, we obtain significant run time improvements by employing fast matrix multiplication.
\end{abstract}

\newpage
\section{Introduction}

Feedback Vertex Set (FVS) is one of the most fundamental NP-complete problems; for example, it was among Karp's original 21 problems~\cite{DBLP:conf/coco/Karp72}.
In FVS we are given an undirected graph $G$ and integer $k$, and are asked whether there exists a set $F$ such that $G[V \setminus F]$ is a forest (i.e.\ $F$ intersects all cycles of $G$).
In the realm of parameterized complexity, where we aim for algorithms with running times of the type $\Os(f(k))$\footnote{The $\Os()$ notation omits factors polynomial in $n$.} with $f(k)$ as small as possible (albeit exponential), FVS is clearly one of the most central problems: To quote~\cite{DBLP:conf/soda/Cao18}, to date the number of parameterized algorithms for FVS published in the literature exceeds the number of parameterized algorithms for any other single problem.

There are several reasons why FVS is the one of the most central problem in parameterized complexity:
First and foremost, the main point of parameterized complexity, being that in many instance the parameter $k$ is small, is very applicable for FVS: In the instances arising from e.g.\ resolving deadlocks in systems of processors~\cite{DBLP:journals/siamcomp/Bar-YehudaGNR98}, or from Bayesian inference or constraint satisfaction, one is only interested in whether small FVS's exist~\cite{DBLP:journals/jair/BeckerBG00, DBLP:journals/ai/Dechter90, DBLP:journals/jacm/WangLS85}.
Second, FVS is a very natural graph modification problems (remove/add few vertices/edges to make the graph satisfy a certain property) that serves as excellent starting point for many other graph modification problems such a planarization or treewidth-deletion (see e.g.~\cite{DBLP:journals/corr/abs-1804-01366} for a recent overview).
Third, FVS and many of its variants (see e.g.~\cite{DBLP:conf/soda/KimK18}) admit elegant duality theorems such as the Erd\"os-P\'osa property; understanding their use in designing algorithms can be instrumental to solve many problems different from FVS faster.
The popularity of FVS also led to work on a broad spectrum of its variations such as Subset, Group, Connected, Simultaneous, or Independent FVS (see for example~\cite{DBLP:conf/iwpec/AgrawalGSS16} and the references therein).

In this paper we study the most basic setting concerning the parameterized complexity of FVS, and aim to design an algorithm with runtime $\Os(f(k))$ with $f(k)$ as small as possible.

One motivation for this study is that we want to get a better insight into the fine-grained complexity of computational problems: How hard is FVS really to solve in the worst-case setting?
Can the current algorithms still be improved significantly or are they close to some computational barrier implied by some hypothesis or conjecture such as, for example, the Strong Exponential Time Hypothesis?

A second motivation is that, lowering the exponential factor $f(k)$ of the running time is a logical first step towards more practical algorithms.
For example, the vertex cover problem\footnote{Given a graph $G$ and integer $k$, find $k$ vertices of $G$ that intersect every edge of $G$.} can be solved in $O(1.28^k + kn)$ time~\cite{DBLP:journals/tcs/ChenKX10},  and a similar running time for FVS would be entirely consistent with our current knowledge. Algorithms with such run times likely outperform other algorithms for a wide variety of instances from practice.
Note there already has been considerable interest in practical algorithms for FVS as it was the subject of the first Parameterized Algorithms and Computational Experiments Challenge (PACE, see e.g.~\cite{DBLP:conf/iwpec/DellHJKKR16}).

For a third motivation of such a study, experience shows an improvement of the running time algorithms for well-studied benchmark problems as FVS naturally goes hand in hand with important new algorithmic tools: The \emph{`race'} for the fastest algorithm for FVS and its variants gave rise to important techniques in parameterized complexity such as Iterative Compression~\cite{DBLP:journals/mst/DehneFLRS07, DBLP:journals/jcss/GuoGHNW06, DBLP:journals/orl/ReedSV04}, Randomized Branching~\cite{DBLP:journals/jair/BeckerBG00} and Cut\&Count~\cite{cutandcount-arxiv}.

\paragraph{The race for the fastest FVS algorithm.}
The aforementioned `race' (see Figure~\ref{fig:race}) started in the early days of parameterized complexity (see e.g~\cite{DBLP:conf/focs/AbrahamsonEFM89}) with an $\Os((2k+1)^k)$ time deterministic algorithm by Downey and Fellows~\cite{DBLP:conf/dagstuhl/DowneyF92}.
We briefly discuss four relevant results from this race.
A substantial improvement of the algorithm from~\cite{DBLP:conf/dagstuhl/DowneyF92} to an $\Os(4^k)$ time randomized algorithm was obtained by Becker et al.~\cite{DBLP:journals/jair/BeckerBG00}.
Their simple but powerful idea is to argue that, if some simple reduction rules do not apply, a random `probabilistic branching' procedure works well.
A few years later, in~\cite{DBLP:journals/mst/DehneFLRS07, DBLP:journals/jcss/GuoGHNW06} it was shown how to obtain $\Os(10.6^k)$ time in the deterministic regime using \emph{Iterative Compression}.
This technique allows the algorithm to assume a feedback vertex set of size $k+1$ is given, which turns out to be useful for detecting feedback vertex sets of size $k$.
The race however stagnated with the paper that introduced the Cut\&Count technique~\cite{cutandcount-arxiv} and gave a $\Os(3^k)$ time randomized algorithm.
In particular, the Cut\&Count technique gave a $\Os(3^\tw)$ time algorithm for FVS if a \emph{tree decomposition} (see Section~\ref{sec:prel} for definitions) of width \tw{} is given, and this assumption can be made due to the iterative compression technique.
After this result, no progress on randomized algorithms for FVS was made as it seemed that improvements over the $\Os(3^\tw)$ running time were not within reach:
In~\cite{cutandcount-arxiv} it was also proven that any $\Os((3-\eps)^{\tw})$ time algorithm, for some $\eps>0$, would violate the SETH.
It was therefore natural to expect the base $3$ is also optimal for the parameterization by the solution size $k$.
Moreover, the very similar $\Os(2^k)$ time algorithm from~\cite{cutandcount-arxiv} for the Connected Vertex Cover problem was shown to be optimal under the Set Cover Conjecture~\cite{DBLP:journals/talg/CyganDLMNOPSW16}.

\begin{figure}
        \center
        \begin{tabular}{|l|l|l|l|}
                \hline
                \textbf{Reference} & \textbf{Running Time} & \textbf{Deterministic?} & \textbf{Year} \\
                \hline
                \hline
                Downey and Fellows~\cite{DBLP:conf/dagstuhl/DowneyF92} & $\Os((2k+1)^k)$ & YES &  1992 \\
                \hline
                Bodlaender~\cite{DBLP:journals/ijfcs/Bodlaender94} & $\Os(17(k^4)!)$ & YES & 1994 \\
                \hline
                Becker et al.~\cite{DBLP:journals/jair/BeckerBG00} & $\Os(4^k)$ & \hfill NO &  2000 \\             
                \hline          
                Raman et al.~\cite{DBLP:conf/isaac/RamanSS02} & $\Os(12^k + (4 \log k)^k)$ &YES &  2002 \\
                \hline
                Kanj et al.~\cite{DBLP:conf/iwpec/KanjPS04} & $\Os((2 \log k + 2 \log \log k + 18)^k )$ &YES & 2004 \\
                \hline
                Raman et al.~\cite{DBLP:journals/talg/RamanSS06} & $\Os((12 \log k/ \log \log k + 6)^k)$ &YES & 2006 \\                               
                \hline
                Guo et al.~\cite{DBLP:journals/jcss/GuoGHNW06}  & $\Os(37.7^k)$ &YES &  2006 \\
                \hline
                Dehne et al.~\cite{DBLP:journals/mst/DehneFLRS07} & $\Os(10.6^k)$ &YES & 2007 \\          
                \hline
                Chen et al.~\cite{DBLP:journals/jcss/ChenFLLV08} & $\Os(5^k)$ &YES &  2008 \\
                \hline
                Cao et al.~\cite{DBLP:journals/algorithmica/CaoC015} & $\Os(3.83^k)$ & YES & 2010 \\         
                \hline
                Cygan et al.~\cite{cutandcount-arxiv} & $\Os(3^k)$ &\hfill NO &  2011 \\              
                \hline
                Kociumaka and Pilipczuk~\cite{DBLP:journals/ipl/KociumakaP14} & $\Os(3.62^k)$ &YES & 2014 \\          
                \hline          
                this paper & $\Os(2.7^k)$, or $\Os(2.6252^k)$ if $\om=2$ &\hfill NO & 2020 \\            
                \hline          
        \end{tabular}
        \caption{The `race' for the fastest parameterized algorithm for Feedback Vertex Set.}
        \label{fig:race}
\end{figure}

\paragraph{Our contributions.}
We show that, somewhat surprisingly, the $\Os(3^k)$ time Cut\&Count algorithm for FVS can be improved:
\begin{restatable}{theorem}{mainfaster}\label{thm:main-faster}
There is a randomized algorithm that solves FVS in time $\Os(2.69998^k)$. If $\om=2$, then the algorithm takes time $\Os(2.6252^k)$.
\end{restatable}

Here $2 \leq \om \leq 2.373$ is the smallest number such that two $n$ by $n$ matrices can be multiplied in $O(n^{\om})$ time~\cite{DBLP:conf/issac/Gall14a}.
Theorem~\ref{thm:main-faster} solves a natural open problem stated explicitly in previous literature~\cite{openprobs}.

Using the method from~\cite{DBLP:conf/stoc/FominGLS16} that transforms $\Os(c^k)$ time algorithms for FVS into $\Os((2-1/c)^n)$ we directly obtain the following improvement over the previously fastest $\Os(1.67^n)$ time algorithm:

\BC
There is a randomized algorithm that solves FVS on an $n$-vertex graph in time $\Os(1.6297^n)$.
\EC

The above algorithms require space exponential in $k$, but we also provide an algorithm using polynomial space at the cost of the running time:

\begin{restatable}{theorem}{main}\label{thm:main}
There is a randomized algorithm that solves FVS in time $\Os(2.8446^k)$ and polynomial space.
\end{restatable}

\paragraph{Our Techniques.}
We build upon the $\Os(3^\tw)$ time algorithm from~\cite{cutandcount-arxiv}.
The starting standard observation is that a feedback vertex set of size $k$ (which we can assume to be known to us by the iterative compression technique) gives a tree decomposition of treewidth $k+1$ with very special properties.
We show how to leverage these properties using the additional assumption that the average degree of all vertices in the feedback vertex set is constant:

\begin{restatable}{lemma}{improvedtw}\label{lem:improved-tw}
        Let $G$ be a graph and $F$ be a feedback vertex set of $G$ of size at most $k$, and define $\bar d := \deg(F)/k = \sum_{v\in F}\deg(v)/k$.
        There is an algorithm that, given $G$ and $F$, computes a tree decomposition of $G$ of width at most $(1-2^{-\bar d}+o(1))k$, and runs in polynomial time in expectation.
\end{restatable}
To the best of our knowledge, Lemma~\ref{lem:improved-tw} is new even for the special case where $F$ is a vertex cover of $G$.
We expect this result to be useful for other problems parameterized by the feedback vertex set or vertex cover size (such parameterizations are studied in for example~\cite{DBLP:conf/ciac/JaffkeJ17}).
Lemma~\ref{lem:improved-tw} is proven via an application of the probabilistic method analyzed via proper colorings in a dependency graph of low average degree. It is presented in more detail in \Cref{sec:twandseps}.

Lemma~\ref{lem:improved-tw}, combined with the $\Os(3^\tw)$ time algorithm from~\cite{cutandcount-arxiv}, implies that we only need to ensure the feedback vertex set has constant average degree in order to get a $\Os((3-\e)^k)$ time algorithm for some $\e>0$.
To ensure this property, we extend the randomized $\Os(4^k)$ time algorithm of Becker et al.~\cite{DBLP:journals/jair/BeckerBG00}.
The algorithm from~\cite{DBLP:journals/jair/BeckerBG00} first applies a set of reduction rules exhaustively, and then selects a vertex with probability proportional to its degree.\footnote{The sampling is usually described as choosing a random edge and then a random vertex of this chosen edge, which has the same sampling distribution.} They show that this chosen vertex appears in an optimal  feedback vertex set with probability at least $1/4$.
To modify this algorithm, we observe that after applying the reduction rules in~\cite{DBLP:journals/jair/BeckerBG00}, every vertex has degree at least $3$, so one idea is to select vertices with probability proportional to $\deg(v)-3$ instead.\footnote{Let us assume that the graph is not $3$-regular, since if it were, then the feedback vertex set has constant average degree and we could proceed as before.}  
It turns out that if $n\gg k$, then this biases us more towards selecting a vertex in an optimal feedback vertex set $F$.
Indeed, we will show that if $n\ge4k$, then we succeed to select a vertex of $F$ with probability at least $1/2$.
This is much better than even success probability $1/3$, which is what we need to beat to improve the $\Os(3^k)$ running time.

Closer analysis of this process shows that even if $n<4k$, as long as the graph itself has large enough average degree, then we also get success probability $\gg 1/3$.
It follows that if the $\deg(v)-3$ sampling does not give success probability $\gg 1/3$, then the graph has $n\le4k$ and constant average degree. Therefore, the graph has only $O(k)$ edges, and even if all of them are incident to the feedback vertex set of size $k$, the feedback vertex set still has constant average degree.
Therefore, we can apply \Cref{lem:improved-tw}, which gives us a modest improvement of the $\Os(3^k)$ running time to $\Os(3^{(1-2^{-56})k})$ time.

\medskip

To obtain improvements to a $\Os(2.8446^k)$ time and polynomial space algorithm, we introduce the new case $n\ll3k$, where we simply add a random vertex to the FVS $F$, which clearly succeeds with probability $\gg1/3$. We then refine our analysis and apply the Cut\&Count method from the $\Os(3^{\tw})$ algorithm in a way similar to~\cite[Theorem~B.1]{cutandcount-arxiv}.

To obtain Theorem~\ref{thm:main-faster} and further improve the above running times, we extend the proof behind Lemma~\ref{lem:improved-tw} to decompose the graph using a ``three-way separation'' (see Definition~\ref{def:threewaysep}) and leverage such a decomposition by combining the Cut\&Count method with fast matrix multiplication. This idea to improve the running time is loosely inspired by previous approaches for MAX-SAT~\cite{DBLP:conf/sat/ChenS15} and connectivity problems parameterized by branch-width~\cite{DBLP:conf/iwpec/PinoBR16}.

\paragraph{Paper Organization.} This paper is organized as follows: We first define notation and list preliminaries in Section~\ref{sec:prel}.
We present the proof of Lemma~\ref{lem:improved-tw} in Section~\ref{sec:twandseps}.
In Section~\ref{sec:sample}, we introduce a probabilistic reduction rule and its analysis.
Subsequently we focus on improving the $\Os(3^k)$ time algorithm for FVS in Section~\ref{sec:3minuseps}.
The algorithm presented there only obtains a modest improvement, but illustrates our main ideas and uses previous results as a black box.

In the second half of the paper we show how to further improve our algorithms and prove our main theorems: Section~\ref{sec:improved} proves Theorem~\ref{thm:main}, and in Section~\ref{sec:mm} we prove \Cref{thm:main-faster}.
Both these sections rely on rather technical extensions of the Cut\&Count method that we postpone to Section~\ref{sec:cc} to improve readability.

\section{Preliminaries}\label{sec:prel}
Let $G$ be an undirected graph. For a vertex $v$ in $G$, $\deg(v)$ is the degree of $v$ in $G$, and for a set $S$ of vertices, we define $\deg(S):=\sum_{v\in S}\deg(v)$. If $S,T \subseteq V(G)$ we denote $E[S,T]$ for all edges intersecting both $S,T$, and denote $E[S]=E[T,T]$.
For a set $\binom{A}{\cdot,\cdot,\cdot}$ denotes all partitions of $A$ into three subsets.
As we only briefly use tree-decompositions we refer to~\cite[Chapter~7]{DBLP:books/sp/CyganFKLMPPS15} for its definitions and standard terminology.

\paragraph{Randomized Algorithms.}
All algorithms in this paper will be randomized algorithms for search problems with one-sided error-probability.
The \emph{(success) probability} of such an algorithm is the probability it will output the asked solution, if it exists.
In this paper we define \emph{with high probability} to be probability at least $1-2^{-c|x|}$ for some large $c$ where $x$ is the input, instead of the usual $1-1/|x|^c$. This is because FPT algorithms take more than simply $\poly(|x|)=O^\star(1)$ time, so a probability bound of $1-2^{-c|x|}$ is more convenient when using an union bound to bound the probability any execution of the algorithm will fail.


Note that if the algorithm has constant success probability, we can always boost it to high probability using $\Os(1)$ independent trials.
For convenience, we record the folklore observation that this even works for algorithms with expected running time:
\begin{lemma}[Folklore]\label{obs:boost}
        If a problem can be solved with success probability $1/S$ and in expected time $T$, and its solutions can be verified for correctness in polynomial time, then it can be also solved in $\Os(S\cdot T)$ time with high probability.
\end{lemma}
\begin{proof}
Consider $cS|x|$ independent runs of the algorithm for some large constant $c$, and if a run outputs a solution, we then verify that solution and output YES if this is successful.
Given that a solution exists, it is not found and verified in any of $cS|x|$ rounds with probability at most $(1-1/S)^{c\cd S|x|} \leq \exp(-cn)$.
The expected running time of the $cS|x|$ independent runs is $c|x|ST$, and by Markov's inequality these jointly run in at most $2c|x|ST$ time with probability at least $3/4$.
Therefore we can terminate our algorithm after $2c|x|ST$ time and by a union bound this gives and algorithm that solves the problem with constant success probability. To boost this success probability to high probability, simply use $|x|$ independent runs of the algorithm that reaches constant success probability.
\end{proof}
        
Using this lemma, we assume that all randomized algorithms with constant positive success probability actually solve their respective problems with high probability.

\paragraph{Separations.} The following notion will be instrumental in our algorithms.
\BD[Separation]
Given a graph $G=(V,E)$, a partition $(A,B,S) \in \bn{V(G)}{\cd,\cd,\cd}$ of $V$ is a \emph{separation} if there are no edges between $A$ and $B$.
\ED


\paragraph{Reduction Rules.}
In the context of parameterized complexity, a \emph{reduction rule} (for FVS) is a polynomial-time transformation of an input instance $(G,k)$ into a different instance $(G',k')$ such that $G$ has a FVS of size $k$ iff $G'$ has a FVS of size $k'$. We state below the standard reduction rules for FVS, as described in \cite{DBLP:books/sp/CyganFKLMPPS15}, Section 3.3. For simplicity, we group all four of their reduction rules FVS.1 to FVS.4 into a single one.

\BR[\cite{DBLP:books/sp/CyganFKLMPPS15}, folklore]\label{red:1} Apply the following rules exhaustively, until the remaining graph has no loops, only edges of multiplicity at most $2$, and minimum vertex degree at least $3$:
\BE
\im If there is a loop at a vertex $v$, delete $v$ from the graph and decrease $k$ by $1$; add $v$ to the output FVS.
\im If there is an edge of multiplicity larger than $2$, reduce its multiplicity to $2$.
\im If there is a vertex $v$ of degree at most $1$, delete $v$.
\im If there is a vertex $v$ of degree $2$, delete $v$ and connect its two neighbors by a new edge.
\EE
\ER

\section{Treewidth and Separators}
\label{sec:twandseps}



In this section, we show how to convert an FVS with small average degree into a good tree decomposition. In particular, suppose graph $G$ has a FVS $F$ of size $k$ with $\deg(F) \le \bar d k$, where $\bar d=O(1)$. We show how to construct a tree decomposition of width $(1-\Om(1))k$. Note that a tree decomposition of width $k+1$ is trivial: since $G-F$ is a forest, we can take a tree decomposition of $G-F$ of width $1$ and add $F$ to each bag. To achieve treewidth $(1-\Om(1))k$, we will crucially use the fact that $\bar d=O(1)$.

We make the assumption that the algorithm already knows the small average degree FVS $F$. This reasoning may seem circular at first glance: after all, the whole task is finding the FVS in the first place. Nevertheless, we later show how to remove this assumption using the standard technique of \emph{Iterative Compression}.

We now present a high level outline of our approach. Our goal is to compute a small set $S$ of vertices---one of size at most $(1-\Om(1))k$---whose deletion leaves a graph of small enough treewidth. Then, taking the tree decomposition of $G-S$ and adding $S$ to each bag gives the desired tree decomposition. Of course, settling for $|S|=(1+o(1))k$ and treewidth $1$ is easy: simply set $S=F$ so that the remaining graph is a forest, which has treewidth $1$. Therefore, it is important that $|S|=(1-\Om(1))k$.


We now proceed with our method of constructing $S$. First, temporarily remove the FVS $F$ from the graph, leaving a forest $T$. We first select a set $S_\e$ of $\be$ vertices to remove from the forest, for some $\be=o(k)$, to break it into connected components such that the edges between $F$ and $T$ are evenly split among the components. More precisely, we want every connected component of $T-S_\e$ to share at most a $1/\be$ fraction of all edges between $F$ and $T$; we show in \Cref{lem:tree} below that this is always possible. The $\be$ vertices in $S_\e$ will eventually go into every bag in the decomposition; this only increases the treewidth by $o(k)$, which is negligible. Hence, we can safely ignore the set $S_\e$.

Next, we perform a \emph{random coloring procedure} as follows: randomly color every connected component of $T-S_\e$ red or blue, uniformly and independently. Let $A$ be the union of all components colored red, and $B$ be the union of all components colored blue. For simplicity of exposition, we will assume here (\emph{with} loss of generality) that $F$ is an independent set: that is, there are no edges between vertices in the FVS. Then, if a vertex $v\in F$ has all its neighbors in $T-S_\e$ belonging to red components, then $v$ only has neighbors in $A$, so let us add $v$ to $A$. Similarly, if all neighbors belong to blue components, then $v$ only has neighbors in $B$, so let us add $v$ to $B$. Observe that the new graphs $G[A]$ and $G[B]$ still have no edges between them, so every vertex addition so far has been ``safe''.

What is the probability that a vertex in $F$ joins $A$ or $B$? Recall that $d(F) = \bar d k$, and since $F$ is an independent set, $|E[F,T-S_\e]| \le |E[F,T]| = d(F) = \bar d k$. If a vertex in $F$ has exactly $\bar d$ edges to $T-S_\e$, then it has probability at least $2^{-\bar d}$ of joining $A$, with equality when all of these edges go to different connected components in $T-S_\e$. Of course, we only have that vertices in $F$ have at most $\bar d$ neighbors on average, but a convexity argument shows that in expectation, at least a $(2^{-\bar d}-o(1))k$ fraction of vertices in $F$ join $A$. That is, $\E[|A\cap F|]\ge(2^{-\bar d}-o(1))k$.  We can make a symmetric argument for vertices joining $B$. Of course, we need both events---enough vertices joining each of $A$ and $B$---to hold simultaneously, which we handle with a concentration argument. From here, it is straightforward to finish the treewidth construction. We now present the formal proofs.


We begin with the following standard fact on balanced separators of forests:

\BL\label{lem:tree}
Given a forest $T$ on $n$ vertices with vertex weights $w(v)$, for any $\beta>0$, we can delete a set $S$ of $\beta $ vertices so that every connected component of $T-S$ has total weight at most $w(V)/\beta$.
\EL
\BP
Root every component of the forest $T$ at an arbitrary vertex. Iteratively select a vertex $v$ of maximal depth whose subtree has total weight more than $w(V)/\beta$, and then remove $v$ and its subtree. The subtrees rooted at the children of $v$ have total weight at most $w(V)/\beta$, since otherwise, $v$ would not satisfy the maximal depth condition. Moreover, by removing the subtree rooted at $v$, we remove at least $w(V)/\beta$ total weight, and this can only happen $\be$ times.
\EP

\BL[Small Separator]\label{lem:sep}
Given an instance $(G,k)$ and a FVS $F$ of $G$ of size at most $k$, define $\bar d := \deg(F)/k$, and suppose that $\bar d=O(1)$. 
There is a randomized algorithm running in expected polynomial time that computes a separation $(A,B,S)$ of $G$ such that:
\BE
\im $|A\cap F|,|B\cap F|\ge(2^{-\bar d}-o(1))k$
\im $|S|\le(1+o(1))k - |A\cap F| - |B\cap F|$
\EE
\EL

\BP

Fix a parameter $\e:=k^{-0.01}$ throughout the proof. Apply \Cref{lem:tree} to the forest $G-F$ with parameter $\e k$, with vertex $v$ weighted by $|E[v,F]|$, and let $S_\e$ be the output. Observe that
\[|S_\e| \le\e k=o(k) ,\]
and every connected component $C$ of $G-F-S_\e$ satisfies
\[ |E[C,F]|\le \f{|E[\bar F,F]|}{\e k} \le \f{\deg(F)}{\e k} = \f{\bar dk}{\e k} = \bar d/\e .\]

Now form a bipartite graph $H$ on vertex bipartition $F\uplus R$, where $F$ is the FVS, and there are two types of vertices in $R$, the \emph{component} vertices and the \emph{subdivision} vertices. For every connected component $C$ in $G-F-S_\e$, there is a component vertex $v_C$ in $R$ that represents that component, and it is connected to all vertices in $F$ adjacent to at least one vertex in $C$. For every edge $e=(u,v)$ in $E[F]$, there is a vertex $v_e$ in $R$ with $u$ and $v$ as its neighbors. Observe that (1) $|R| \le |E[\bar F,F]| + 2|E[F]| = \deg(F)$, (2) every vertex in $R$ has degree at most $\bar d/\e$, and (3) the degree of a vertex $v\in F$ in $H$ is at most $\deg(v)$.

The algorithm that finds a separator works as follows. For each vertex in $R$, color it red or blue uniformly and independently at random. Every component $C$ in $G-F-S_\e$ whose vertex $v_C$ is colored red is added to $A$ in the separation $(A,B,S)$, and every component whose vertex $v_C$ is colored blue is added to $B$. Every vertex in $F$ whose neighbors are all colored red joins $A$, and every vertex in $F$ whose neighbors are all colored blue joins $B$. The remaining vertices in $F$, along with the vertices in $S_\e$, comprise $S$.

\BSCL\label{clm:sep1}
$(A,B,S)$ is a separation.
\ESCL
\begin{subproof}[]
        Suppose for contradiction that there is an edge connecting $A$ and $B$. The edge cannot connect two distinct components of $G-F-S_\e$, so it must have an endpoint in $F$. The edge cannot connect a vertex in $F$ to a vertex in $G-F-S_\e$, since a vertex in $F$ only joins $A$ or $B$ if all of its neighbors in $R$ are colored the corresponding color. Therefore, the edge $e$ must connect two vertices in $F$. But then, $v_e$ connects to both endpoints and is colored either red or blue, so it is impossible for one endpoint of $e$ to have all neighbors colored red, and the other endpoint to have all neighbors colored blue, contradiction.
\end{subproof}

We now show that with good probability both Conditions (1) and (2) hold. The algorithm can then repeat the process until both conditions hold.


\BSCL\label{clm:cond1}
With probability at least $1-1/\poly(k)$, Condition (1) holds for $(A,B,S)$.
\ESCL
\begin{subproof}[]
        There are at most $\e|F|$ vertices in $F$ with degree at least $\overline d/\e$. Since they cannot affect condition (1) by an additive $\e|F|\le \e  k=o(k)$ factor, we can simply ignore them; let $F'$ be the vertices with degree at most $\bar d/\e$. Consider the \emph{intersection graph} $I$ on the vertices of $F'$, formed by connecting two vertices in $F'$ iff they share a common neighbor (in $R$). Since every vertex in $F'$ and $C$ has degree at most $\bar d/\e$, the maximum degree of $I$ is $(\bar d/\e)^2$. Using the standard greedy algorithm, we color $F'$ with $(\bar d/\e)^2+1$ colors so that every color class forms an independent set in $I$. In particular, within each color class, the outcome of each vertex---namely, whether it joins $A$ or $B$ or $S$---is independent across vertices.
        
        
        Let $F'_i$ be the vertices colored $i$. If $|F_i'| < k^{0.9}$, then ignore it; since $\bar d\le O(1)$ and $\e=k^{-0.01}$, the sum of all such $|F'_i|$ is at most $((\bar d/\e)^2+1)k^{0.9}=o(k)$, so they only affect condition (1) by an additive $o(k)$ factor. Henceforth, assume that $|F'_i|\ge k^{0.9}$. Each vertex $v\in F'_i$ has at most $\deg(v)$ neighbors in $H$, so it has independent probability at least $2^{-\deg(v)}$ of joining $A$. Let $X_i:=|F'_i\cap A|$ be the number of vertices in $F_i'$ that join $A$; by Hoeffding's inequality\footnote{If $a_1,\ldots,a_n$ are independent and Bernoulli and $X=a_1+a_2+\ldots+a_n$, then $\Pr[|X-E[x]| \geq t] \leq 2\exp(-2t^2/n)$.},
        \begin{align*}
         \Pr[ X_i \le E[X]-k^{0.8} ] &\le 2 \exp ( -2 \cd (k^{0.8})^2/ |F'_i| )\\
         &\le 2\exp(-2 \cd k^{0.6}) \le 1/\poly(k)
        \end{align*} 
        for large enough $k$.
        
        By a union bound over all $\le k^{0.1}$ color classes $F'_i$ with $|F'_i| \geq k^{0.9}$, the probability that $|F'_i\cap A|\ge E[|F'_i\cap A|]-k^{0.8}$ for each $F'_i$ is $1-1/\poly(k)$. In this case,
        \begin{align*}
        |F\cap A| &\ge \sum_{i: |F'_i| \geq k^{0.9}}\lp E[|F'_i\cap A|]-k^{0.8}\rp
        \\&\ge \sum_{i: |F'_i| \geq k^{0.9}}\sum_{v\in F'_i} 2^{-\deg(v)}-k^{0.1} \cd k^{0.8}
        \\&= \sum_{v\in F'}2^{-\deg(v)}-o(k)
        \\&\ge |F'|\cd 2^{-\deg(F')/|F'|}-o(k)
        ,\end{align*}
        where the last inequality follows from convexity of the function $2^{-x}$. Recall that $|F'|\ge(1-o(1))k$, and observe that $\deg(F')/|F'| \le \deg(F)/|F|=\bar d$ since the vertices in $F\bs F'$ are precisely those with degree exceeding some threshold. It 
        \[ |F\cap A| \ge (1-o(1))k \cd 2^{-\bar d} ,\]
        proving condition (1) for $|A\cap F|$. Of course, the argument for $|B\cap F|$ is symmetric.
 \end{subproof}

\BSCL
With probability at least $1-1/\poly(k)$, Condition (2) holds for $(A,B,S)$.
\ESCL
\begin{subproof}[]
        At most $\e k=o(k)$ vertices in $S$ can come from $S_\e$, and the other vertices in $S$ must be precisely $F \bs ((A\cap F) \cup (B\cap F))$, which has size $k-|A\cap F|-|B\cap F|$. 
\end{subproof}
Hence, with at least constant probability, both Conditions~(1)~and~(2) hold. Furthermore, whether or not they hold can be checked in polynomial time, so the algorithm can repeatedly run the algorithm until the separation satisfies both conditions.
\EP

\improvedtw*
\BP
Compute a separation $(A,B,S)$ following \Cref{lem:sep}. 
Since $(A\cap F) \cup S$ is a FVS of $A\cup S$ of size $(1-2^{-\bar d}+o(1))k$, we can compute a tree decomposition of $G[(A\cap F)\cup S]$ of width $(1-2^{-\bar d}+o(1))k$ as follows: start with a tree decomposition of width $1$ of the forest $G[(A\cap F)\cup S]-(F\cup S)$, and then add all vertices in $(A\cap F)\cup S$ to each bag. Similarly, compute a tree decomposition of $G[(B\cap F)\cup S]$ in the same way. Finally, merge the two tree decompositions by adding an edge between an arbitrary node from each decomposition; since there is no edge connecting $A$ to $B$, the result is a valid tree decomposition.
\EP


\section{Probabilistic Reduction}\label{sec:sample}

Whenever a reduction fails with a certain probability, we call it a \emph{probabilistic reduction}. Our probabilistic reduction is inspired by the randomized $\Os(4^k)$ FVS algorithm of \cite{DBLP:journals/jair/BeckerBG00}.
Whenever we introduce a probabilistic reduction, we include (P) in the header, such as in the reduction below.

\BR[P]\label{red:sample}
Assume that \Cref{red:1} does not apply and $G$ has a vertex of degree at least $4$. Sample a vertex $v\in V$ proportional to $w(v):=(\deg(v)-3)$. That is, select each vertex $v$ with probability $w(v)/w(V)$. Delete $v$, decrease $k$ by $1$.
\ER
We say a probabilistic reduction \emph{succeeds} if it selects a vertex in an optimal feedback vertex set.
\BO\label{obs:deg}
Let $G$ be a graph $F$ a FVS of $G$. Denoting $\bar{F}:= V \setminus F$ we have that
\begin{gather} \deg(\bar F) \le \deg(F) + 2(|\bar F|-1) .\label{1}\end{gather}
\EO
\BP
Since $G-F$ is a forest, there can be at most $|\bar F|-1$ edges in $G-F$, each of which contributes $2$ to the summation $\deg(\bar F)=\sum_{v\in\bar F}\deg(v)$. The only other edges contributing to $\deg(\bar F)$ are in $E[F,\bar F]$, which contribute $1$ to both $\deg(\bar F)$ and $\deg(F)$. Therefore,
\[ \deg(\bar F) \le 2(|\bar F|-1) + |E[F,\bar F]| \le 2(|\bar F|-1) + \deg(F) .\]
\EP

\BL\label{lem:sample}
If $n\ge 4k$ and the instance is feasible, then \Cref{red:sample} succeeds with probability at least $1/2$.
\EL
\BP
Let $F\s V$ be a FVS of size $k$.\footnote{From any FVS of size less than $k$, we can arbitrarily add vertices until it has size $k$.} We show that the probability of selecting a vertex in $F$ is at least $1/2$. Define $\bar F := V\bs F$, so that our goal is equivalent to showing that $w(F) \ge w(\bar F)$.

The value of $w(F)$ can be rewritten as
\begin{gather} w(F)=\sum_{v\in F}(\deg(v)-3)= \deg(F)-3|F|.\label{2} \end{gather}
By \Cref{obs:deg}, 
\begin{gather} w(\bar F)=\sum_{v\in \bar F}(\deg(v)-3) = \deg(\bar F)-3|\bar F| \stackrel{(\ref1)}\le \deg(F) + 2(|\bar F|-1) - 3|\bar F| \le \deg(F)-|\bar F| . \label{3}\end{gather}
Therefore,
\begin{align}
w(F)\ge w(\bar F) &\impliedby \deg(F)-3|F| \ge \deg(F)-|\bar F| \label{4'}
\\&\iff |\bar F|\ge 3|F| \nonumber
\\&\iff n\ge4k. \nonumber
\end{align}
\EP

Therefore, as long as $n\ge4k$, we can repeatedly apply \Cref{red:1,red:sample} until either $k=0$, which means we have succeeded with probability at least $1/2^k$, or we have an instance $(G,k)$ with $n\le 4k$.

Later on, we will need the following bound based on the number of edges $m$. Informally, it says that as long as the average degree is large enough, \Cref{red:sample} will still succeed with probability close to $1/2$ (even if $n < 4k$).
\BL\label{lem:samplem}
Assume that $2m>3n$.
If the instance is feasible, then \Cref{red:sample} succeeds with probability at least $\min\{\frac12,\f{m-n-2k}{2m-3n}\}$.
\EL
\BP
There are at most $|\bar F|-1$ edges not contributing to $\deg(F)$, so
\begin{gather}
 m \le (|\bar F|-1) + \deg(F) \le (n-k) + \deg(F) \implies \deg(F) \ge m-n+k.  \label{5''}
\end{gather}
If $w(F)/w(\bar F)\ge1$, then the success probability is at least $1/2$, so assume otherwise that $w(F)< w(\bar F)$. Following the proof of \Cref{lem:sample}, the contrapositive of (\ref{4'}) gives
\begin{gather}
w(F)< w(\bar F) \implies |\bar F|< 3|F| , \label{5'}
\end{gather}
so we have
\[ \f{w(F)}{w(\bar F)} \stackrel{(\ref2)}= \f{\deg(F)-3|F|}{w(\bar F)} \stackrel{(\ref3)}\ge  \f{\deg(F)-3|F|}{\deg( F)-|\bar F|}\stackrel{(\ref{5''},\,\ref{5'})}\ge \f{(m-n+k)-3|F|}{(m-n+k)-|\bar F|} = \f{m-n-2k}{m-2n+2k}.\]
Finally, as the Lemma statement is vacuous when $2k > m-n$, the Lemma follows.
\EP

\section{$\Os((3-\e)^k)$ Time Algorithm}\label{sec:3minuseps}

In this section we present our simplest algorithm that achieves a running time of $\Os((3-\e)^k)$, for some $\e>0$.
The improvement $\e$ is very small, but we found this to be the simplest exposition that achieves the bound for any $\e>0$.
We build on the following result:

\begin{lemma}[Cygan et al.~\cite{cutandcount-arxiv}]\label{lem:treewidthdp}
        There is an algorithm \treewidthDP that, given a tree decomposition of the input graph of width $\tw$, and parameter $k$ outputs a FVS of size at most $k$ with high probability if it exists. Moreover, the algorithm runs in $\Os(3^\tw)$ time.
\end{lemma}

First, we combine the tree decomposition from the previous section with the standard technique of \emph{Iterative Compression} to build an algorithm that runs in time $\Os((3-\e)^k)$ time, assuming that $m=O(k)$ (recall $m$ denotes the number of edges of the input graph).
Then, we argue that by applying \Cref{red:sample} whenever $m\gg k$, we can essentially ``reduce'' to the case $m=O(k)$. Combining these two ideas gives us the $\Os((3-\e)^k)$ algorithm.

The algorithm is introduced below in pseudocode. The iterative compression framework proceeds as follows. We start with the empty graph, and add the vertices of $G$ one by one, while always maintaining a FVS of size at most $k$ in the current graph. Maintaining a FVS of the current graph allows us to use the small tree decomposition procedure of \Cref{sec:twandseps}. Then, we add the next vertex in the ordering to each bag in the tree decomposition, and then solve for a new FVS in $O^\star(3^\tw)$ time using \Cref{lem:treewidthdp}. Of course, if there is no FVS of size $k$ in the new graph, then there is no such FVS in $G$ either, so the algorithm can terminate early.

\begin{algorithm}[H]
\mylabel{alg:ic}{\texttt{IC1}}
\caption{\ref{alg:ic}$(G,k)$}
\small
\textbf{Input}: Graph $G=(V,E)$ and parameter $k$, with $m=O(k)$. \\
\textbf{Output}: FVS $F$ of size at most $k$, or \Inf if none exists. 
\begin{algorithmic}[1]
\State Order the vertices $V$ arbitrarily as $(v_1,\lds,v_n)$\label{line:ic1}
\State $F\gets\emptyset$
\For {$i=1,\lds,n$} \Comment{\textbf{Invariant:} $F$ is a FVS of $G[\{v_1,\lds,v_{i-1}\}]$}
 \State Compute a tree decomposition of $G[\{v_1,\lds,v_{i-1}\}]$ by applying \Cref{lem:improved-tw} on input $F$\label{line:ic4}
 \State Add $v_i$ to each bag in the tree decomposition\label{line:ic5}
 \State $F\gets $ a FVS of $G[\{v_1,\lds,v_i\}]$ with parameter $k$, computed using \treewidthDP\label{line:ic6} from Lemma~\ref{lem:treewidthdp}
 \If{$F$ is \Inf}
  \State \Return \Inf
 \EndIf
\EndFor
\State \Return $F$
\end{algorithmic}
\end{algorithm}

\BL\label{cor:tw-alg}
On input instance $(G,k)$ with $m=O(k)$, $\ref{alg:ic}(G,k)$ runs in time  $\Os(3^{(1-2^{-2m/k}+o(1))k})$. Moreover, if there exists a FVS $F$ of size at most $k$, then \ref{alg:ic} will return a FVS of size at most $k$ with high probability.
\EL
\BP
Suppose that there exists a FVS $F^*$ of size at most $k$. Let $(v_1,\lds,v_n)$ be the ordering from Line~\ref{line:ic1}, and define $V_i:=\{v_1,\lds,v_i\}$. Observe that $F^*\cap V_i$ is a FVS of $G[V_i]$, so the FVS problem on Line~\ref{line:ic6} is feasible. By Lemma~\ref{lem:treewidthdp}, Line~\ref{line:ic6} correctly computes a FVS with high probability on any given iteration. Therefore, after using $O^*(1)$ independent trials,  with high probability a FVS is returned successfully.

We now bound the running time. On Line~\ref{line:ic4}, the current set $F$ is a FVS of $G[V_{i-1}]$. To bound the value of $\bar d$ used in \Cref{lem:improved-tw}, we use the (rather crude) bound
\begin{gather*}
 \deg(F) \le \deg(V) = 2m \implies \bar d=\f{\deg(F)}k \le \f{2m}k , 
\end{gather*}
 and moreover, $\bar d=O(1)$ since $m=O(k)$ by assumption. Therefore, \Cref{lem:improved-tw} guarantees a tree decomposition of width at most $(1-2^{-2m/k}+o(1))k$, and adding $v_i$ to each bag on Line~\ref{line:ic5} increases the width by at most $1$. By Lemma~\ref{lem:treewidthdp}, Line~\ref{line:ic6} runs in time $\Os(3^{(1-2^{-2m/k}+o(1))k})$ time, as desired.
\EP

We now claim below that if $m\ge\Om(k)$ for a sufficiently large $k$, then \Cref{red:sample} succeeds with good probability (in particular, with probability greater than $1/3$).
\BL\label{lem:m28k}
If $G$ has a FVS of size $k$ and $m\ge28k$, then
\Cref{red:sample} succeeds with probability at least $4/11$.
\EL
\BP
We consider two cases.
If $n \geq 4k$, then the success probability is at least $1/2$ by \Cref{lem:sample}. Otherwise, if $n\le4k$, then $m\ge28k\ge7n$, and \Cref{lem:samplem} and the trivial bound $k\le n$ give a success probability of at least
\[ \f{m-n-2k}{2m-3n} \ge \f{m-3n}{2m-3n} \ge \f{7n-3n}{14n-3n}=\f4{11} .\]
Hence, regardless of whether or not $n\ge4k$, \Cref{red:sample} succeeds with probability at least $4/11$.
\EP

Below is the full randomized algorithm in pseudocode, which combines \Cref{red:1,red:sample} with the iterative compression routine \ref{alg:ic}.
After a trivial check and reduction rule, Line~\ref{line:coinflip} flips a coin that needs to be flipped \Heads in order to proceed to the iterative compression step.

The motivation for this is that we want each iteration of \ref{alg:fvs1t} to run quickly in expectation---in particular, in $\Os(3^{o(k)})$ time---for simplicity of analysis.
This way, if the algorithm has success probability $c^{-k}$ for some constant $c$, then we can repeat it $\Os(c^k)$ times, succeeding with high probability and taking $\Os(c^{(1+o(1))k})$ time in expectation. Since \ref{alg:ic} takes $\Os(3^{(1-2^{-56}+o(1))k})$ time by \Cref{cor:tw-alg}, we should call \ref{alg:ic} with probability at most $3^{-(1-2^{-56})k}$, which is exactly the probability of the coin flipping \Heads.

\begin{algorithm}[H]
\mylabel{alg:fvs1t}{\texttt{FVS1}}
\caption{\ref{alg:fvs1t}$(G,k)$}
\small
\textbf{Input}: Graph $G=(V,E)$ and parameter $k\le n$. \\
\textbf{Output}: A FVS of size $k$ with probability $3^{-(1-2^{-56})}$ if one exists; \Inf otherwise. 
\begin{algorithmic}[1]
\LineIf{$k=0$}{\Return $\emptyset$ if $G$ is acyclic, and \Return \Inf otherwise}
  \State Exhaustively apply \Cref{red:1} to $(G,k)$ to get vertex set $F$ and instance $(G',k')$ with $m'$ edges
  \State Flip a coin with \Heads probability $3^{-(1-2^{-56})k'}$ \label{line:coinflip}
  \If {$m'\le28k'$ and coin flipped \Heads}
    \State $F'\gets \ref{alg:ic}(G',k')$ \label{line:ic}
  \Else
    \State Apply \Cref{red:sample} to $(G',k')$ to get vertex $v\in V$ and instance $(G'',k'-1)$
    \State $F'\gets \ref{alg:fvs1t}(G'',k'-1)\cup\{v\}$\Comment{$\Inf\cup S=\Inf$ for any set $S$} \label{line:rec1}
  \EndIf
  \State \Return $F\cup F'$
\end{algorithmic}
\end{algorithm}


\BL\label{clm:time1}
$\ref{alg:fvs1t}(G,k)$ runs in expected $\Os(3^{o(k)})$ time and has $\Omega(3^{-(1-2^{-56})k})$ success probability.
\EL
\BP
For the running time, 
the computation outside of Line~\ref{line:ic} clearly takes $\poly(n)$ time. For each $k'\in(k_0,k]$, Line~\ref{line:ic} is executed with probability $3^{-(1-2^{-56})k'}$ and takes $\Os(3^{(1-2^{-56}+o(1))k'})$ time, so in expectation, the total computation cost of Line~\ref{line:ic} is $\Os(2^{o(k)})$ per value of $k'$, and also $\Os(2^{o(k)})$ overall.

It remains to lower bound the success probability.
Define $c:=3^{1-2^{-56}}$. We will prove by induction on $k$ that $\ref{alg:fvs1t}(G,k)$ succeeds with probability at least $c^{-k}/2$.
This statement is trivial for $k=0$, since no probabilistic reductions are used and $\ref{alg:fvs1t}(G,k)$ succeeds with probability $1$. For the inductive step, consider an instance $\ref{alg:fvs1t}(G,k+1)$. First, suppose that $m\le28k$. In this case, if \ref{alg:ic} in Line~\ref{line:ic} is executed, then it will run in time $\Os(3^{(1-2^{-2m/k}+o(1))k})$ by \Cref{cor:tw-alg}, and correctly output a FVS $F$ of size at most $k$, with high probability. This happens with probability at least
\[ 3^{-(1-2^{-56})k} \cd \lp 1-\f1{\poly(n)}\rp \ge c^{-k} \cd \f12, \]
as desired. If \ref{alg:ic} is not executed, then \ref{alg:fvs1t} can still succeed, but this only increases our overall success probability, so we disregard it.

Otherwise, suppose that $m>28k$. Then, by \Cref{lem:m28k}, applying \Cref{red:sample} succeeds with probability at least $4/11$. By induction, the recursive call on Line~\ref{line:rec1} succeeds with probability at least $c^{-(k-1)}/2$, so the overall probability of success is at least
\[ \f4{11} \cd\f{c ^{-(k-1)}}{2} \ge c^{-1} \cd\f{c ^{-(k-1)}}{2} = \f{c^{-k}}2 ,\]
as desired.
\EP

The claimed $\Os((3-\eps)^{k})$ time algorithm follows from Lemma~\ref{clm:time1} by boosting the success probability of Algorithm~\ref{alg:fvs1t} according to Lemma~\ref{obs:boost}.

\section{Improved Algorithm and Polynomial Space}
\label{sec:improved}

In this section, we present the $\Os(2.8446^k)$ time algorithm promised by \Cref{thm:main}. At a high level, our goal is to obtain a tighter bound on $\bar d=\deg(F)/k$, which we only bounded loosely by $2m/k$ in \Cref{sec:3minuseps}.  Recall that the treewidth bound of $(1-2^{-\bar d}+o(1))k$ from \Cref{lem:improved-tw} has exponentially dependence on $\bar d$, so every constant factor savings in $\bar d$ is crucial.

 First, we introduce another simple reduction step, which works well when $n \ll 3k$.

\BR[P]\label{red:uniform}
Sample a uniformly random vertex $v$. Delete $v$ and decrease $k$ by $1$.
\ER
For the entire section, we will fix a constant $\e>0$ and obtain a running time that depends on $\e$. At the very end, we will optimize for $\e$ and achieve the running time $\Os(2.8446^k)$. For formality, we define the following assumption (\ref{as:1}) and state the corresponding direct claim.
\begin{gather}
n\le(3-\e)k\tag{A1}\label{as:1}
\end{gather}

\BCL\label{clm:a1}
If (\ref{as:1}) is true, then \Cref{red:uniform} succeeds with probability at least $1/(3-\e)$.
\ECL

Now suppose that (\ref{as:1}) is false. Observe that \Cref{red:sample} succeeds with probability at least $1/(3-\e)$ precisely when
\[\f{w(F)}{w(\bar F)} \stackrel{(\ref2)}= \f{\deg(F)-3|F|}{\deg(\bar F)-3|\bar F|} \ge \f1{2-\e}.\]
By \Cref{obs:deg}, we have
\[ \f{\deg(F)-3|F|}{\deg(\bar F)-3|\bar F|} \stackrel{(\ref1)}\ge \f{\deg(F)-3|F|}{(\deg(F)+2|\bar F|)-3|\bar F|} = \f{\deg(F)-3k}{\deg(F)-(n-k)}, \] 
and since (\ref{as:1}) is false, 
\[ \f{\deg(F)-3k}{\deg(F)-(n-k)} \ge \f{\deg(F)-3k}{\deg(F)-((3-\e)k-k)} = \f{\deg(F)-3k}{\deg(F)-(2-\e)k} .\]
We are interested in whether or not
\[ \f{\deg(F)-3k}{\deg(F)-(2-\e)k} \stackrel?\ge \f1{2-\e} \iff (2-\e)(\deg(F)-3k)\stackrel?\ge\deg(F)-(2-\e)k \iff \deg(F)\stackrel?\ge\f{4-2\e}{1-\e}k ,\]
which, if true, would imply that \Cref{red:sample} succeeds with probability at least $1/(3-\e)$. Again, we present the assumption and corresponding claim:
\begin{gather}
 \deg(F)\ge\f{4-2\e}{1-\e}k \text{\qquad for some FVS $F$ of size $k$} \tag{A2}\label{as:2}
\end{gather}

\BCL\label{clm:a2}
If (\ref{as:1}) is false and (\ref{as:2}) is true, then \Cref{red:sample} succeeds with probability at least $1/(3-\e)$.
\ECL

An immediate issue in this assumption is that the algorithm does not know $\deg(F)$, so it cannot determine whether (\ref{as:2}) is true or not. This can be accomplished by designing an algorithm to find Feedback Vertex Sets with additional properties defined as follows:
\newcommand{\BFVS}{BFVS\xspace}
\BD[Bounded Total Degree FVS]
In the \emph{bounded total degree FVS (\BFVS)} problem, the input is an unweighted, undirected graph $G$ on $n$ vertices, and parameters $k\le n$ and $\bar d\le O(1)$. The goal is to either output a FVS $F$ of size at most $k$ satisfying $\deg(F)\le\bar dk$, or correctly conclude none exists. \ED

\begin{algorithm}[H]
\mylabel{alg:mw}{\texttt{IC2}}
\caption{\ref{alg:mw}$(G,k,\bar d)$}
\small
\textbf{Input}: Graph $G=(V,E)$ and parameters $k\le n$ and $\bar d=O(1)$. \\
\textbf{Output}: A FVS $F$ of size at most $k$ satisfying $\deg(F)\le \bar dk$, or \Inf if none exists. 

\begin{algorithmic}[1]
\State Order the vertices $V$ arbitrarily as $(v_1,\lds,v_n)$\label{line:btd1}
\State $F\gets\emptyset$
\For {$i=1,\lds,n$} \Comment{\textbf{Invariant:} $\deg_{}(F) \le \bar dk$}
  \State Compute a separation $(A,B,S')$ of $G[\{v_1,\ldots,v_{i-1}\}]$ by Lemma~\ref{lem:sep} on input $F$\label{line:btd4}
  \State $S\gets S'\cup\{v_i\}$, so that $(A,B,S)$ is a separation of $G[\{v_1,\ldots,v_{i}\}]$ \label{line:btd5}
  \State $F\gets \ref{alg:btd-count}(G[\{v_1,\lds,v_i\}],k+1,A,B,S)$ \label{line:btd6}
  \If {$F$ is \Inf}
    \State \Return \Inf
  \EndIf
\EndFor
\State \Return $F$
\end{algorithmic}
\end{algorithm}

We remark that Lines~\ref{line:btd5}~and~\ref{line:btd6} replace the tree decomposition and \treewidthDP of \ref{alg:ic}.
Indeed. we need to solve the BFVS problem instead of FVS, and \treewidthDP could be easily extended to solve this problem as well.
However, it \treewidthDP crucially relies on exponential working space.
In the new algorithm we circumvent this by exploiting special properties of the separation directly. The function of the new algorithm is described by the following lemma:

\BL\label{lem:btd-count}
        There is an Algorithm~\ref{alg:btd-count} that, given $G$, a FVS $F$ of $G$ of size $k$, parameter $\bar d$, and a separation $(A,B,S)$ as given by Lemma~\ref{lem:sep},
        outputs a FVS of size at most $k-1$ satisfying $\deg(F)\le \bar d(k-1)$, or \Inf if none exists. The algorithm uses $\Os(3^{(1-2^{-\bar d}+o(1))k})$ time and polynomial space.
\EL
Because of its technical nature, we postpone the proof of this Lemma to Subsection~\ref{subsec:ccsimplesep}.

\BL\label{lem:btd-time}
        Algorithm~\ref{alg:mw} solves the \BFVS problem in $\Os(3^{(1-2^{-\bar d}+o(1))k})$ time and polynomial space.
\EL
\BP
Suppose that there exists a FVS $F^*$ of size at most $k$ satisfying $\deg(F^*)\le\bar dk$. Let $(v_1,\lds,v_n)$ be the ordering from Line~\ref{line:btd1}, and define $V_i:=\{v_1,\lds,v_i\}$. Observe that $F^*\cap V_i$ is a FVS of $G[V_i]$ satisfying $\deg(F^*\cap V_i)\le\bar dk$, so the FVS problem on Line~\ref{line:btd6} is feasible. By Lemma~\ref{lem:btd-count}, Line~\ref{line:btd6} correctly computes a FVS with high probability on any given iteration. Therefore, with high probability, a FVS is returned successfully by a union bound.

We now bound the running time. On Line~\ref{line:btd4}, the current set $F$ is a FVS of $G[V_{i-1}]$ satisfying $\deg(F)\le\bar dk$, so \Cref{lem:improved-tw} guarantees a tree decomposition of width at most $(1-2^{-\bar d}+o(1))k$, and adding $v_i$ to each bag on Line~\ref{line:btd5} increases the width by at most $1$. By Lemma~\ref{lem:btd-count}, Line~\ref{line:btd6} runs in time $\Os(3^{(1-2^{-\bar d}+o(1))k})$ time, as desired.
Lastly, the space bound follows clearly from the descriptions of \ref{alg:mw} and Lemma~\ref{lem:btd-count}. \EP

\begin{algorithm}[H]
\mylabel{alg:fvst}{\texttt{FVS2}}
\caption{\ref{alg:fvst}$(G,k)$}
\small
\textbf{Input}: Graph $G=(V,E)$ and parameter $k\le n$. \\
\textbf{Output}: Either output a FVS $F$ of size $k$, or (possibly incorrectly) conclude that one does not exist (\Inf). 

\begin{algorithmic}[1]
\LineIf{$k=0$}{\Return $\emptyset$ if $G$ is acyclic, and \Return \Inf otherwise}

    \State Exhaustively apply \Cref{red:1} to $(G,k)$ to get vertex set $F$ and instance $(G',k')$\label{line:apply-red}
    \State $\bar d\gets (4-2\e)/(1-\e)$
    \State Flip a coin with \Heads probability $3^{-(1-2^{-\bar d})k'}$
    \If {coin flipped \Heads}
      \State $F'\gets \ref{alg:mw}(G',k',\bar d)$ \label{line:mw}
    \Else
      \If {$n' \le (3-\e)k'$}\Comment{(\ref{as:1}) is true}
        \State Apply \Cref{red:uniform} to $(G',k')$ to get vertex $v\in V$ and instance $(G'',k'-1)$
      \Else\Comment{(\ref{as:1}) is false}
        \State Apply \Cref{red:sample} to $(G',k')$ to get vertex $v\in V$ and instance $(G'',k'-1)$
      \EndIf
      \State $F'\gets \ref{alg:fvst}(G'',k'-1)\cup\{v\}$\Comment{Denoting $\Inf\cup S=\Inf$ for any set $S$} \label{line:rec}
    \EndIf
  \State \Return $F\cup F'$
\end{algorithmic}
\end{algorithm}

\BL\label{lem:polyalgo}
Fix the parameter $\e \in (0,1)$, and let $c_\e:=\max \{ 3-\e,  3^{1-2^{-(4-2\e)/(1-\e)}}\}$.
If $c_\e\ge2$, then $\ref{alg:fvst}(G,k)$ succeeds with probability at least $c_\e^{-k}/k$.
Moreover, Algorithm $\ref{alg:fvst}(G,k)$ has $\Os(3^{o(k)})$ expected running time.
\EL
\BP
For the running time, the computation outside of Line~\ref{line:mw} clearly takes $\poly(n)$ time. For each $k'\in(k_0,k]$, Line~\ref{line:mw} is executed with probability $3^{-(1-2^{-\bar d})k'}$, and takes $\Os(3^{(1-2^{-\bar d}+o(1))k'})$ time by \Cref{lem:btd-time}. Therefore, in expectation, the total computation cost of Line~\ref{line:mw} is polynomial per value of $k'$, and also polynomial overall.

We continue with proving by induction on $k$ that $\ref{alg:fvst}(G,k)$ succeeds with probability at least $c^{-k}/k$ (we denote $c:=c_\e$).
This statement is trivial for $k=0$, since no probabilistic reductions are used and $\ref{alg:fvst}(G,k)$ succeeds with probability $1$. For the inductive step, consider an instance $\ref{alg:fvst}(G,k+1)$. Let $(G',k')$ be the reduced instance after Line~\ref{line:apply-red}. First, suppose that (\ref{as:2}) is false on instance $(G',k')$. That is, every FVS $F$ of size at most $k$ satisfies $\deg(F)\le \f{4-2\e}{1-\e}k' $; here, we will only need the existence of one such $F$. In this case, if \ref{alg:mw} in Line~\ref{line:mw} is executed, then it will correctly output a FVS $F$ of size at most $k$, with high probability by \Cref{lem:btd-time}. This happens with probability at least
\[ 3^{-(1-2^{-\bar d})k'} \cd \lp 1-\f1{\poly(n)}\rp \ge c^{-k'} \cd \f1k \ge \f{c^{-k}}k, \]
as desired.

Otherwise, suppose that (\ref{as:2}) is true on instance $(G',k')$. Then, by \Cref{clm:a1,clm:a2}, regardless of whether (\ref{as:1}) is true, the reduction applied succeeds with probability at least $1/(3-\e)$. This is assuming, of course, that Line~\ref{line:mw} is not executed, which happens with probability $1-c^{-k'}\ge 1-2^{-k'} \ge 1-1/k'$ since $c\ge2$. By induction, the recursive call on Line~\ref{line:rec} succeeds with probability at least $c^{-(k'-1)}/(k'-1) $, so the overall probability of success is at least
\[ \lp1-\f1{k'}\rp \cd \f1{3-\e} \cd\f{c ^{-(k'-1)}}{k'-1} \ge \lp1-\f1{k'}\rp \cd \f1c \cd \f{c ^{-(k'-1)}}{k'-1} = \f{c^{-k'}}{k'}\ge \f{c^{-k}}k ,\]
as desired.
\EP

To optimize for $c_\e$, we set $\e\approx 0.155433$, giving $c_\e\le 2.8446$. \Cref{thm:main} now follows by combining Lemma~\ref{lem:polyalgo} with Lemma~\ref{obs:boost}.

\section{Further Improvement Using Matrix Multiplication}
\label{sec:mm}

In this section, we further speed up the algorithm \ref{alg:mw} that solves the BFVS problem. First, we open the Cut\&Count black box, which essentially transforms the FVS (or BFVS) problem to counting the number of partitions of the graph that satisfy a particular constraint, modulo some integer. The transformation are similar to the presentation in~\cite{cutandcount-arxiv}, so we defer the details to \Cref{sec:cc}. In~\cite{cutandcount-arxiv}, this counting problem is solved using dynamic programming on a tree decomposition in $\Os(3^\tw)$ time, which can be translated to an $\Os(3^k)$ time algorithm for BFVS. 

As with most problems efficiently solvable on tree decompositions, the Cut\&Count problem performs well when given small vertex separators. Indeed, we show in Subsection~\ref{subsec:ccsimplesep} that instead of calling the $\Os(3^{\tw})$ algorithm on the tree decomposition from \Cref{lem:improved-tw}, we can solve the problem by applying dynamic programming on the $(A,B,S)$ separation from \Cref{lem:sep} directly in the same running time, and also in polynomial space. The resulting algorithm is the algorithm \ref{alg:btd-count} promised by \Cref{lem:btd-count}.

How do we obtain an even faster running time, then? The main insight in this section is that the counting problem has a special arithmetic nature that also makes it amenable to \emph{matrix multiplication} as well. Combining these two observations, we construct a \emph{three-way vertex separation} of the graph $G$, defined as follows:

\BD[Three-Way Separation]\label{def:threewaysep}
Given a graph $G=(V,E)$, a partition $(S_1,S_2,S_3,S_{1,2},S_{1,3},S_{2,3},S_{1,2,3})$ of $V$ is a \emph{separation} if there are no edges between any two sets $S_I,S_J$ whose sets $I$ and $J$ are disjoint.
\ED

The construction of a good three-way separation is very similar to the ``two-way separation'' in \Cref{lem:improved-tw}: it also features a randomized coloring procedure and is proven using concentration arguments. We then apply a combination of dynamic programming and matrix multiplication on the three-way separator, which is presented as Algorithm~\ref{alg:mm} in Subsection~\ref{subsec:cc3waysep}.

\subsection{Three-Way Separator}

\BL[Three-Way Separator]\label{lem:3way}
Given an instance $(G,k)$ and a FVS $F$ of size at most $k$, define $\bar d := \deg(F)/k$, and suppose that $\bar d=O(1)$. 
There is a polynomial time algorithm that computes a three-way separation $(S_1,S_2,S_3,S_{1,2},S_{1,3},S_{2,3},S_{1,2,3})$ of $G$ such that there exists values $f_1,f_2$ satisfying:
  \BE
  \im[1a.] $f_1\ge3^{-\bar d}$ 
  \im[1b.] $(f_1-o(1))k\le|S_i\cap F|\le(f_1+o(1))k$ for all $i\in[3]$
  \im[2a.] $f_2\ge (2/3)^{\bar d}-2f_1$ 
  \im[2b.] $(f_2-o(1))k\le|S_{i,j}\cap F|\le(f_2+o(1))k$ for all $1\le i<j\le3$
  \EE
\EL
\BP
Our proof follows the outline of the proof of \Cref{lem:sep}. Initially, we start out the same: fix $\e:=k^{-0.01}$, apply \Cref{lem:tree} on the same input (that is, $G-F$), and construct the bipartite graph $H$ on the bipartition $F\uplus R$ in the same manner as in \Cref{lem:sep}. We recall that (1) $|R| \le |E[\bar F,F]| + 2|E[F]| = \deg(F)$, (2) every vertex in $R$ has degree at most $\bar d/\e$, and (3) the degree of a vertex $v\in F$ in $H$ is at most $\deg(v)$.

Now, instead of randomly two-coloring the vertex set $R$, the algorithm three-colors it. That is, for each vertex in $R$, color it with a color in $\{1,2,3\}$ chosen uniformly and independently at random. For each subset $I\s 2^{[3]}\bs\{\emptyset\}$, create a vertex set $S_I$ consisting of all vertices $v\in F$ whose neighborhood in $H$ sees the color set $I$ precisely. More formally, let $c(v)$ and $N(v)$ be the color of $v\in R$ and the neighbors of $v$ in $H$, and define
$ S_I = \{ v\in R : \bigcup_{u\in N(v)}c(u) = I \} .$
Furthermore, if $I$ is a singleton set $\{i\}$, then add (to $S_I$) all vertices in the connected components $C$ whose component vertex in $R$ is colored $i$. From now on, we abuse notation, sometimes referring to sets  $S_{\{1\}}$, $S_{\{1,2\}}$, etc.\ as  $S_1$, $S_{1,2}$, etc.

The proof of the following easy Subclaim is essentially the same as the proof of~\Cref{clm:sep1} (but with more cases), and therefore omitted.


\BSCL
$(S_1,S_2,S_3,S_{1,2},S_{1,3},S_{2,3},S_{1,2,3})$ is a three-way separation.
\ESCL

We start by proving Conditions (1a) and (1b) with the following strategy. First, we first present a value $f_1$ such that Condition (1b) holds with probability $1-1/\poly(k)$. Then, we argue that actually, this value of $f_1$ satisfies Condition (1a) (with probability $1$). 

\BSCL\label{clm:cond1b}
For $f_1:=\lp\sum_dp_d\cd|F'_d|\rp/|F'|$, Condition (1b) holds with  probability $1-1/\poly(k)$.
\ESCL
\begin{subproof}

The proof uses similar concentration arguments as the proof of \Cref{clm:cond1}. Again, fix a parameter $\e:=k^{-0.01}$ throughout the proof. Let $F'$ be the vertices with degree at most $\bar d/\e$, so that again, $|F'|\ge(1-o(1))|F|$. Form the intersection graph $I$ on the vertex set $F'$ as in \Cref{clm:cond1}, and color it with $(\bar d/\e)^2+1$ colors with a standard greedy algorithm.

Let $F'_d$ be the vertices in $F'$ with degree $d\le\bar d/\e$ in $H$, and let $F'_{i,d}$ be the vertices colored $i$ with degree $d$ in $H$. If $|F'_{i,d}|<k^{0.9}$, then ignore it; since $\bar d\le O(1)$ and $\e=k^{-0.01}$, the sum of all such $F'_{i,d}$ is at most $((\bar d/\e)^2+1) \cd (\bar d/\e)\cd k^{0.9} = o(k)$, so they only affect condition (1b) by an additive $o(k)$ factor. Henceforth, assume that $|F'_i|\ge k^{0.9}$. 

We only focus our attention on $S_1$; the claim for $S_2$ and $S_3$ are identical. The probability that a vertex $v\in F'_{i,d}$ joins $S_1$ is a fixed number $p_d$ that only depends on $d$. Let $X:=|F'_{i,d}\cap A|$ be the number of vertices in $F'_{i,d}$ that join $S_1$; we have $\E[X]=p_d\cd|F'_{i,d}|$, and by Hoeffding's inequality,
\[ \Pr[|X-\E[X]|\ge k^{0.8}]\le2\exp(-2\cd(k^{0.8})^2/|F'_{i,d}|)\le2\exp(-2\cd k^{1.6}/k)\le1/\poly(k) \]
for large enough $k$. Taking a union bound over all colors $i$ and degrees $d$, we conclude that with probability $1-1/\poly(k)$,
\[ \left| |F'\cap S_1| - \E \lb |F'\cap S_1| \rb \right| \le ((\bar d/\e)^2+1)\cd (\bar d/\e)\cd k^{0.8} + o(k) = o(k) .\]
Moreover,
\[ \E[|F'\cap S_1|] = \sum_d p_d\cd|F'_d|, \]
and we see that
\[ ||S_1\cap F|-f_1k| = ||S_1\cap F'|-f_1\cd|F'||+o(k)=o(k) ,\]
which fulfills condition (1b).
\end{subproof}

\BSCL
For $f_1:=\lp\sum_dp_d\cd|F'_d|\rp/|F'|$, Condition (1a) holds with  probability $1-1/\poly(k)$.
\ESCL
\begin{subproof}
The number $p_d$ equals $3^{-d}$, so $f_1=(\sum_d|F'_d|\cd3^{-d})/|F'|$. Observe that $\deg(F')/|F'|\le\deg(F)/|F|=\bar d$, since the vertices in $F\bs F'$ are precisely those with degree exceeding some threshold. Therefore,
\begin{align*}
f_1 &= \f1{|F'|}\sum_d|F'_d|\cd3^{-d}
\\&= \f1{|F'|}\sum_{v\in F'}3^{-\deg(v)}
\\&\ge 3^{-\deg(F')/|F'|},
\end{align*}
where the last inequality follows from convexity of the function $3^{-x}$.
\end{subproof}

\BSCL\label{clm:cond2b}
For $f_2 := \lp\sum_dp_d\cd|F'_d|\rp/k$, Condition (2b) holds with probability $1-1/\poly(k)$.
\ESCL
\begin{subproof}
The proof is identical to that of \Cref{clm:cond1b}, except that $p_d$ is now the probability that a vertex $v\in F'_{i,d}$ joins $S_2$.
\end{subproof}

\BSCL
The $f_1:=\lp\sum_dp_d\cd|F'_d|\rp/|F'|$ and $f_2 := \lp\sum_dp_d\cd|F'_d|\rp/k$, condition (2a) holds.
\ESCL
\begin{subproof}
Here, our strategy is slightly different. Let $q_d$ be the probability that a vertex $v$ of degree $d$ in $H$ joins one of $S_1$, $S_2$, and $S_{1,2}$. Since this is also the probability that no neighbor of $v$ is colored $3$, we have $q_d=(2/3)^d$. Let $p_{1,d}$ and $p_{2,d}$ be the value of $p_d$ in the proofs of \Cref{clm:cond1b} and \Cref{clm:cond2b}, respectively, so that $q_d = 2p_{1,d}+p_{2,d}$. Therefore,
\begin{align*}
2f_1+f_2 &= 2\cd\f1{|F'|}\sum_dp_{1,d}\cd|F'_d| + \f1{|F'|}\sum_dp_{2,d}\cd|F'_d|
\\&= \f1{|F'|}\sum_d q_d\cd|F'_d|
\\&= \f1{|F'|}\sum_d|F'_d|\cd\lp\f23\rp^{d}
\\&= \f1{|F'|}\sum_{v\in F'}\lp\f23\rp^{\deg(v)}
\\&\ge \lp\f23\rp^{\deg(F')/|F'|},
\end{align*}
where the last inequality follows from convexity of the function $(2/3)^x$. Again, we have $\deg(F')/|F'|\le\deg(F)/|F|=\bar d$, so
\[ f_2 \ge \lp\f23\rp^{\deg(F')/|F'|} -2f_1\ge \lp\f23\rp^{\bar d}-2f_1 ,\]
which fulfills condition (2a).
\end{subproof}

\EP

\subsection{Matrix Multiplication Algorithm}

In this section, we present the improved iterative compression algorithm \ref{alg:btd-mm}. It is mostly unchanged from \ref{alg:mw}, except that the algorithm now computes a three-way separator and calls the faster BFVS algorithm~\ref{alg:mm} on it. Because of its technical nature, the algorithm \ref{alg:mm} and its analysis are deferred to Subsection~\ref{subsec:cc3waysep}. Instead, we simply state its running time guarantee in \Cref{lem:mmalgo} below.

\begin{algorithm}[H]
\mylabel{alg:btd-mm}{\texttt{IC3}}
\caption{\ref{alg:btd-mm}$(G,k,\bar d)$}
\small
\textbf{Input}: Graph $G=(V,E)$ and parameters $k\le n$ and $\bar d=O(1)$. \\
\textbf{Output}: A FVS $F$ of size at most $k$ satisfying $\deg(F)\le \bar dk$, or \Inf if none exists. 
\begin{algorithmic}[1]
\State Order the vertices $V$ arbitrarily as $(v_1,\lds,v_n)$
\State $F\gets\emptyset$
\For {$i=1,\lds,n$} \Comment{\textbf{Invariant:} $\deg_{}(F) \le \bar dk$}
  \State Compute a separation $(S_1,S_2,S_3,S_{1,2},S_{1,3},S_{2,3},S_{1,2,3})$ of $G[\{v_1,\ldots,v_{i-1}\}]$ by Lemma~\ref{lem:3way} on input $F$
  \State $S_{1,2,3}\gets S_{1,2,3}\cup\{v_i\}$, so that $(S_1,S_2,S_3,S_{1,2},S_{1,3},S_{2,3},S_{1,2,3})$ is a three-way separation of $G[\{v_1,\ldots,v_{i}\}]$
  \State $F\gets \ref{alg:mm}(G[\{v_1,\lds,v_i\}],k+1,S_1,S_2,S_3,S_{1,2},S_{1,3},S_{2,3},S_{1,2,3})$
  \If {$F$ is \Inf}
    \State \Return \Inf
  \EndIf

\EndFor
\State \Return $F$
\end{algorithmic}
\end{algorithm}

\BL\label{lem:mmalgo}
There is an Algorithm~\ref{alg:mm} that, given $G$, a FVS $F$ of $G$ of size $k$, parameter $\bar d$, and a separation $(S_1,S_2,S_3,S_{1,2},S_{1,3},S_{2,3},S_{1,2,3})$ as given by \Cref{lem:3way},
outputs a FVS of size at most $k-1$ satisfying $\deg(F)\le \bar d(k-1)$, or \Inf if none exists.
The algorithm runs in time $\Os(3^{(1-\min\{ (2/3)^{\bar d}, (3-\om)(2/3)^{\bar d} + (2\om-3)3^{-\bar d} \}+o(1))k})$.
\EL

Assuming \Cref{lem:mmalgo}, we prove our main result, \Cref{thm:main-faster}, restated below.
\mainfaster*
\BP
We run \ref{alg:fvst}, replacing every occurrence of \ref{alg:mw} with \ref{alg:btd-mm}. Following \ref{alg:fvst}, we define $\bar d:=(4-2\e)/(1-\e)$ for some $\e>0$ to be determined later; note that $\bar d\ge4$ for any $\e>0$. Since $\om<2.3728639$~\cite{DBLP:conf/issac/Gall14a}, by \Cref{lem:btd-mm}, \ref{alg:btd-mm} runs in time $\Os(3^{(1-((3-\om)(2/3)^{\bar d}+(2\om-3)\cd3^{-\bar d})+o(1))k})$, so \ref{alg:fvst} runs in time $\Os(c_\e^k)$ for $c_\e:=\max \{ 3-\e,  3^{1-((3-\om)(2/3)^{(4-2\e)/(1-\e)}+(2\om-3)\cd3^{-(4-2\e)/(1-\e)})+o(1)}\}$. To optimize for $c_\e$, we set $\e\approx 0.3000237$, giving $c_\e\le2.699977 $.

If $\om=2$, then by \Cref{lem:btd-mm}, \ref{alg:btd-mm} runs in time $\Os(3^{(1-(2/3)^{\bar d}+o(1))k})$, so \ref{alg:fvst} runs in time $\Os(c_\e^k)$ for $c_\e:=\max \{ 3-\e,  3^{1-(2/3)^{(4-2\e)/(1-\e)}+o(1)}\}$. To optimize for $c_\e$, we set $\e\approx 0.3748068$, giving $c_\e\le2.6252 $.
\EP

\section{Cut and Count}
\label{sec:cc}

In this section we open the black box formed by the Cut\&Count approach~\cite{cutandcount-arxiv}.
It should be noted that most of this section, except Subsection~\ref{subsec:cc3waysep}, is very similar to the methods from~\cite{cutandcount-arxiv}.
We need the following definition:

\begin{definition}[\cite{cutandcount-arxiv}]\label{def:cutobjects}
        Let $G$ be a graph with weight function $\omega: V(G) \rightarrow \mathbb{N}$. Let $s,m',W$ be integers. Define $\mathcal{C}^{\omega,s,m'}_W$ to be the set
        \[
        \Big\{(F,L,R) \in \binom{V(G)}{\cdot,\cdot,\cdot} \Bigm| \omega(F)=W \wedge E[L,R]=\emptyset  \wedge |F|=s \wedge |E[L\cup R]|=m'  \Big\}.
        \]
\end{definition}

In the above, $\binom{V(G)}{\cdot,\cdot,\cdot}$ denotes the set of all partitions of $V(G)$ into three sets (denoted by $F$ for `Feedback Vertex Set', $L$ for `left side of the cut', and $R$ for `right side of the cut').
In words, a partition $(F,L,R)$ of the vertex set is an element of $\mathcal{C}^{\omega,s,m'}_W$ if the total weight of all vertices in $F$ equals $W$, there are no edges between $L$ and $R$, exactly $m'$ edges with both endpoints either in $L$ or in $R$, and $|F|=s$.
The use of Definition~\ref{def:cutobjects} becomes clear in the following lemma.
Intuitively, the crux is that $F$ is a FVS of $G$ if and only if for some $s,m',W$ the number of partitions $(L,R)$ of $V \setminus F$ such that $|\mathcal{C}^{\omega,s,m'}_W|$ is odd; in this case $\deg(F)$ can be read off from $W$.

\begin{lemma}\label{lem:forestfvsequiv}
        Let $G$ be a graph and $d$ be an integer. Pick $\omega(v) \in_R \{1,\ldots,2|V|\}$ uniformly and independent at random for every $v\in V(G)$, and define $\omega'(v):=|V|^2\omega(v)+d(v)$. The following statements hold:
        \begin{enumerate}
                \item If for some integers $m'$, $W=i|V|^2+d$ we have that $|\mathcal{C}^{\omega',k,m'}_W|\not\equiv 0\ (\mathrm{mod}\ 2^{n-k-m'})$, then $G$ has a feedback vertex set $F$ satisfying $\deg(F)=d$.
                \item If $G$ has a feedback vertex set $F$ satisfying $\deg(F)=d$, then with probability at least $1/2$ for some $m'$, $W=i|V|^2+d$ we have that $|\mathcal{C}^{\omega',k,m'}_W|\not\equiv 0\ (\mathrm{mod}\ 2^{n-k-m'})$.
        \end{enumerate}
\end{lemma}

Lemma~\ref{lem:forestfvsequiv} states that in order to solve the Feedback Vertex Set problem it is sufficient to compute $|\mathcal{C}^{\omega,n-k,m'}_W|$ for all setting of the parameters. 
Before proving the Lemma we need to recall some standard building blocks:
\begin{lemma}[Lemma~A.7 in~\cite{cutandcount-arxiv}]\label{lem:fvsvscc}
        A graph with $n$ vertices and $m$ edges is a forest iff it has at most $n-m$ connected components.
\end{lemma}

\begin{definition}
        A function $\omega\colon U \rightarrow \mathbb{Z}$ \emph{isolates} a set family $\mathcal{F} \subseteq 2^U$ if there is a unique $S' \in \mathcal{F}$ with $\omega(S')= \min_{S \in \mathcal{F}}\omega(S)$, where~$\omega(S') := \sum _{v \in S'} \omega(v)$.
\end{definition}

\begin{lemma}[Isolation Lemma, \cite{MulmuleyVV87}]
        \label{lem:iso}
        Let $\mathcal{F} \subseteq 2^U$ be a non-empty set family over a universe $U$.
        For each $u \in U$, choose a weight $\omega(u) \in \{1,2,\ldots,W\}$ 
        uniformly and independently at random.
        Then 
        $\Pr[\omega \textnormal{ isolates } \mathcal{F}] \geq 1 - |U|/W$.
\end{lemma}

\begin{proof}[Proof of Lemma~\ref{lem:forestfvsequiv}]
        We first prove 1. Note that if $|\mathcal{C}^{\omega',k,m'}_W|\not\equiv 0\ (\mathrm{mod}\ 2^{n-k-m'})$, there must be some vertex subset $F$ such that the number of choices $L,R$ with $(F,L,R) \in \mathcal{C}^{\omega,s,m'}_W$ is not a multiple of $2^{n-k-m'}$. As we can independently decide for each component of $G[V \setminus F]$ whether to include it in $L,R$ $G[V \setminus F]$ thus must have at most $n-k-m'$ connected components.
        By Lemma~\ref{lem:fvsvscc} it therefore must be a FVS. The condition on the degree follows by the weighting.
        
        For item 2. First apply Lemma~\ref{lem:iso} with $U=V$ and the set family $\mathcal{F}$ being the set of all feedback vertex sets $F$ of $G$ satisfying $\deg(F)=d$.
        With probability $1/2$, there will be some weight $i$ such that there is a unique FVS $F$ with $\deg(F)=d$ of weight $i$.
        By Lemma~\ref{lem:fvsvscc} this is the only $F$ that has a contribution to $|\mathcal{C}^{\omega',k,m'}_W|$ that is not a multiple of $2^{n-k-m'}$ as the number of extension of $F$ to an object of $\mathcal{C}^{\omega',k,m'}_W$ is exactly $2^{\mathrm{cc}(G[V \setminus F])}$,\footnote{Here we let $\mathrm{cc}$ denote the number of connected components} assuming $\omega(F)=W$, $|F|=k$ and $|E[V \setminus F]|=m'$.
\end{proof}

We now continue with a lemma that is useful towards computing $|\mathcal{C}^{\omega,s,m'}_W|$.

\begin{definition}
        If $F^0 \subseteq V(G)$ is a FVS of $G$ and $(F,L,R) \in \binom{F^0}{\cdot,\cdot,\cdot}$, we denote
        \[
        c^{\omega,s,m'}_W(F,L,R) = |\{ (F',L',R') \in \mathcal{C}^{\omega,s,m'}_W: F' \cap F^0 = F \wedge L' \cap F^0= L\wedge R' \cap F^0 = R\}|.
        \]
\end{definition}

\begin{lemma}\label{lem:forestDP}
        There is a polynomial time algorithm $\forestDP(G,\omega,F,L,R,s,m',W)$ that, given a graph $G$, weight function $\omega:V(G) \rightarrow \mathbb{N}$, vertex sets $F,L,R$ and integers $s,k,m',W$ computes $c^{\omega,s,m'}_W(F,L,R)$ in $\poly(n,W)$ time, assuming that $F \cup L \cup R$ is an FVS of $G$.
\end{lemma}
\begin{proof}

We denote $F_0 = F \cup L \cup R$ for the given FVS.
We will use dynamic programming over the forest induced by $V \setminus (F \cup L \cup R)$, in a way very similar to the proof of~\cite[Theorem~B.1]{cutandcount-arxiv}. 
We assign roots to each tree in the forest $V(G) \setminus F_0$ arbitrarily, so the standard relations parents, children, ancestors and descendants are well-defined.
For a vertex $v$, we denote $T[v]$ for the tree induced by $v$ and all its descendants. 
If $v$ has $d$ children (which we order in arbitrary fashion) and $i \leq d$, we also denote $T[v,i]$ for the tree induced by $v$ and all descendants of its first $i$ children.

For $x \in \{L',R',F'\}$, the table entries for the dynamic programming are defined as follows:
\begin{align*}
        A^{(x)}_{W,s,m'}[v,i] := |\{ (F',L',R') \in \binom{V(T[v,i] )\cup F_0}{\cdot,\cdot,\cdot}&: F' \cap F_0 = F \wedge L' \cap F_0 = L \wedge R' \cap F_0 = R \wedge \omega(F')= W\\
                                            &\wedge E[L',R']=\emptyset \wedge |F'|=s \wedge |E[L' \cup R']|=m' \wedge v \in x \}|.
\end{align*}
For convenience we also denote $A^{(x)}_{W,s,m'}[v]$ for $A^{(x)}_{W,s,m'}[v,d]$, where $d$ is number of children of $v$.

If $v$ is a leaf of a tree in the forest $V \setminus F_0$, then it is easy to see that we have
        \[
                A^{(x)}_{W,s,m'}[v,0] =
                \begin{cases}
                1, &\text{if } \omega(v)=W-\omega(F) \wedge |F|=s \wedge |E[L^* \cup R^*]|=m' \wedge E[L^*,R^*]= \emptyset \\
                0, &\text{otherwise},
                \end{cases}     
        \]
        where $L^*$ denotes $L' \cup \{v\}$ if $x= L'$ and $L'$ otherwise, and similarly $R^*$ denotes $R' \cup \{v\}$ if $x=R'$ and $R'$ otherwise.

If $v$ has children $v_1,\ldots,v_d$ in the forest $V \setminus F_0$, we have that
        \begin{align*}
                A^{(F')}_{W,s,m'}[v,1] &=  \sum_{x'\in \{L',R',F'\} } A^{(x')}_{W-\omega(F'),s-1,m'}[v_1]\\
                A^{(L')}_{W,s,m'}[v,1] &=  [N(v) \cap R = \emptyset](A^{(L')}_{W,s,m'-|N(v) \cap L|-1}[v_1]+ A^{(F')}_{W,s,m'-|N(v) \cap L|}[v_1])\\
                A^{(R')}_{W,s,m'}[v,1] &=  [N(v) \cap L = \emptyset](A^{(R')}_{W,s,m'-|N(v) \cap R|-1}[v_1]+ A^{(F')}_{W,s,m'-|N(v) \cap R|}[v_1])\\
        \end{align*}
        Here we use Iverson's bracket notation $[b]$ for a Boolean predicate $b$ which denotes $1$ if $b$ is true and $0$ otherwise.
        
        To see that this holds, note we need to account for the possible contributions of $v$ to $\omega(F')$, $|F'|$ and need to check whether $E[L',R']=\emptyset$ is not violated and account for an increase of $E[L' \cup R']$ which may include the edge $\{v,v_1\}$.
        
        Moreover, for $i>1$ we have that
        
        \begin{align*}
        A^{(F')}_{W,s,m'}[v,i] &=  \sum_{\substack{x' \in \{L',R',F'\} \\ W_1+ W_2 = W - |F| \\ s_1+s_2= m' - |F| \\ m'_1+m_2= m' - |E[L \cup R]|}} A^{(x)}_{W_1,s_1,m'_1}[v,i-1]* A^{(x')}_{W_2,s_2,m'_2}[v_i]. \\
        A^{(L')}_{W,s,m'}[v,i] &=  \sum_{\substack{x' \in \{L',F'\} \\ W_1+ W_2 = W - |F| \\ s_1+s_2= m' - |F| \\ m'_1+m_2= m' - |E[L \cup R]|-[x'=L']}} A^{(x)}_{W_1,s_1,m'_1}[v,i-1]* A^{(x')}_{W_2,s_2,m'_2}[v_i]. \\
        A^{(R')}_{W,s,m'}[v,i] &=  \sum_{\substack{x' \in \{R',F'\} \\ W_1+ W_2 = W - |F| \\ s_1+s_2= m' - |F| \\ m'_1+m_2= m' - |E[L \cup R]|-[x'=R']}} A^{(x)}_{W_1,s_1,m'_1}[v,i-1]* A^{(x')}_{W_2,s_2,m'_2}[v_i].
        \end{align*}
        
        Similarly as before we need account for the possible contributions of $v$ to $\omega(F')$, $|F'|$ and need to check whether $E[L',R']=\emptyset$ is not violated and account for an increase of $E[L' \cup R']$ which may include the edge $\{v,v_1\}$.
        Note we compensate for double counting due to $F,L,R$.
        
        Finally we can merge the counts stored at the roots of each tree of the forest to get the desired value. Specifically, if the the roots are $r_1,\ldots,r_c$ then
        \begin{equation}\label{eq:mergeforest}
                c^{\omega,s,m'}_W(F,L,R) = \sum_{\substack{ x_1,\ldots,x_c \in \{L',R',F\} \\  W_1+ \ldots + W_d = W -(d-1)|F| \\ s_1+\ldots+s_d= m' - (d-1)|F| \\ m'_1+\ldots+m'_d= m' - (d-1)|E[L \cup R]|}} \prod_{i=1}^d A^{(x_i)}_{W_i,s_i,m'_i}[r_i].
        \end{equation}
        Here we again compensate for double counting due to $F,L,R$. 
        Given all entries $A^{(x_i)}_{W_i,s_i,m'_i}[r_i]$, we can combine~\eqref{eq:mergeforest} with standard dynamic programming to compute $c^{\omega,s,m'}_W(F,L,R)$ is polynomial time.
\end{proof}

\subsection{Cut and Count Using Simple Separation: Proof of Lemma~\ref{lem:btd-count}}
\label{subsec:ccsimplesep}
The algorithm promised by the lemma is listed in Algorithm~\ref{alg:btd-count}.
For the claimed time bound, note all steps are polynomial time except Lines~\ref{line:sloop},~\ref{line:aloop} and~\ref{line:bloop}, and these jointly give rise to $3^{(|S|+|A\cap F|)}+3^{(|S|+|B\cap F|)}$ iterations.
As the separation $(A,B,S)$ was assumed to satisfy the properties $|A\cap F|,|B\cap F|\ge(2^{-\bar d}-o(1))k$ and $|S|\le(1-2\cd2^{-\bar d}+o(1))k$ from Lemma~\ref{lem:sep}, the time bound follows.

For the correctness, we claim that at Line~\ref{line:rescheck} $\cnt=|\mathcal{C}^{\omega',k,m'}_W|$ for some $m'$, $W=i|V|^2+d$.
The lemma follows from this by Lemma~\ref{lem:forestfvsequiv}. To see the claim, observe that the algorithm iterates over all partitions $(F,L,R)$ of the separator $S$ in Line~\ref{line:sloop}. For each partition, and every way to split $W,k,m$ (Line~\ref{line:wkmloop}), the algorithm computes the number $\cntA$ (resp.\ $\cntB$) of ``extensions'' of the partition in $G[A\cup S]$ (resp.\ $G[B\cup S]$) that ``respect'' the split, and then multiples $\cntA$ and $\cntB$. To see why the two counts are multiplied, observe that there are no edges between $A$ and $B$ in the separation $(A,B,S)$, so extending into $G[A\cup S]$ is independent to extending into $G[B\cup S]$.


\begin{algorithm}[H]
        \mylabel{alg:btd-count}{\texttt{BFVS1}}
        \caption{\ref{alg:btd-count}$(G,F,k,A,B,S)$}
        \small
        \textbf{Input}: Graph $G=(V,E)$, FVS $F$ of size $k$, parameters $k,\bar d\le n$, and separation $(A,B,S)$ from Lemma~\ref{lem:sep}\\
        \textbf{Output}: A FVS of size at most $k$ satisfying $\deg(F)\le \bar dk$, or \Inf if none exists. 
        \begin{algorithmic}[1]
                \State Pick $\omega \in_R \{1,\ldots,2|V|\}$ uniformly and independently at random for every $v \in V(G)$
                \State Set $\omega'(v):= |V|^2\omega(v)+d(v)$
                \State $\cnt \gets 0$  
                \For{$m',W$ such that $0\leq m'\leq m$, $W = |V^2|i+d \leq \omega(V)$ for some $d\leq \bar dk$}
                \For{$(F_S,L_S,R_S) \in \binom{S}{\cdot,\cdot,\cdot}$}\label{line:sloop}
                \For{$W_1,k_1,m'_1$ such that $0\leq W_1 \leq W, 0\leq k_1 \leq k,0 \leq m'_1 \leq m'$} \label{line:wkmloop}
                \State $\cntA \gets 0$
                \For{$(F_A,L_A,R_A) \in \binom{A \cap F}{\cdot,\cdot,\cdot}$}\label{line:aloop}
                \State $\cntA \gets\cntA+ \forestDP(G[A \cup S],\omega,F_A \cup F_S,L_A \cup L_S,R_A \cup R_S,k_1,m'_1,W_1)$
                \EndFor
                \State $\cntB \gets 0$
                \For{$(F_B,L_B,R_B) \in \binom{B \cap F}{\cdot,\cdot,\cdot}$}\label{line:bloop}
                \State $\cntB \gets\cntB+ \forestDP(G[B \cup S],\omega,F_B \cup F_S,$
                \item[]  $ \hspace{15em}         L_B \cup L_S,R_B \cup R_S,k+|F_S|-k_1,m'+|E[L_S \cup R_S]|-m'_1,W+\omega(F_S)-W_1)$
                \EndFor
                \State $\cnt \gets \cnt + \cntA\cd\cntB$
                \EndFor
                \EndFor
                \EndFor
                \If{$\cnt \not\equiv 0\ (\mathrm{mod}\ 2^{n-k-m'})$} \label{line:rescheck}
                \State \Return a FVS of $G[\{v_1,\lds,v_i\}]$ of size $\le k$ satisfying $\deg(F)\le\bar dk$, constructed by self-reduction
                \EndIf
                \State \Return \Inf
        \end{algorithmic}
\end{algorithm}
\subsection{Cut and Count Using $3$-way Separation: Proof of Lemma~\ref{lem:mmalgo}}
\label{subsec:cc3waysep}

We now present the improved BFVS algorithm below. First, we argue its correctness, that at Line~\ref{line:mm17}, $\cnt = |\m C_W^{\om',k,m'}|$. First, the algorithm iterates over partitions of $S_{1,2,3}$ in Line~\ref{line:mm-for0} the same way Algorithm~\ref{alg:btd-count} iterates over partitions of $S$. The rest of the algorithm, which includes the matrix multiplication routine, seeks to compute the number of extensions in $S_1\cup S_2\cup S_3$ given each partition of $S_{1,2}\cup S_{1,3}\cup S_{2,3}$ (and given the current partition of $S_{1,2,3}$, as well as a three-way split of $W,k,m'$). Like in the case of separator $(A,B,S)$ in \ref{alg:btd-count}, it is true that the extensions of $S_1$, $S_2$, and $S_3$ are independent given the partition of $S_{1,2}\cup S_{1,3}\cup S_{2,3}$, but in this case, the size of $|S_{1,2}\cup S_{1,3}\cup S_{2,3}|$ can be prohibitively large. Instead, to compute this quantity efficiently, first observe that since there are no edges between $S_1$ and $S_{2,3}$, the number of extensions of $S_1$ only depends on the partition of $S_{1,2}\cup S_{1,3}$, and not $S_{2,3}$. For each partition of $S_{1,2}\cup S_{1,3}$, take the graph $H$ defined in Line~\ref{line:mm5}, and imagine adding an edge between the respective partitions of $S_{1,2}$ and $S_{1,3}$, weighted by the number of extensions in $S_1$. We proceed analogously for extensions of $S_2$ and $S_3$. Finally, the total number of extensions (given the partition of $S_{1,2,3}$) amounts to computing, for all triangles in $H$, the product of the weights of the three edges (Line~\ref{line:mm13}), which can be solved by a standard matrix multiplication routine.

Finally, the desired running time bound is more complicated for \ref{alg:mm}. We prove \Cref{lem:btd-mm} below which, together with \Cref{lem:3way}, implies the running time bound of Lemma~\ref{lem:mmalgo}.

\begin{algorithm}[H]
        \mylabel{alg:mm}{\texttt{BFVS2}}
        \caption{\ref{alg:mm}$(G,k,S_1,S_2,S_3,S_{1,2},S_{1,3},S_{2,3},S_{1,2,3})$}
        \small
        \textbf{Input}: Graph $G=(V,E)$, FVS $F$, parameters $k,\bar d\le n$, and separation $(S_1,S_2,S_3,S_{1,2},S_{1,3},S_{2,3},S_{1,2,3})$ from \Cref{lem:3way} \\
        \textbf{Output}: A FVS of size at most $k$ satisfying $\deg(F)\le \bar dk$, or \Inf if none exists. 
        
        \begin{algorithmic}[1]

                \For {$\poly(n)$ iterations}
                \State Pick $\omega \in_R \{1,\ldots,2|V|\}$ uniformly and independently at random for every $v \in V(G)$ 
                \State Set $\omega'(v):= |V|^2\omega(v)+d(v)$
                \State $\cnt\gets0$
                \For{$m',W$ such that $0\leq m'\leq m$, $W = |V^2|i+d \leq \omega(V)$ for some $d\leq \bar dk$}
                \For{$(F_0,L_0,R_0)\in \bn{S_{1,2,3}}{\cd,\cd,\cd}$}\label{line:mm-for0}
                \For{nonnegative $W_i,k_i,m'_i$, $i\in[3]$ such that $\sum_iW_i=W,\ \sum_ik_i=k,\ \sum_im'_i= m'$}
                \State $H\gets$ an empty graph with vertices indexed by $\bn{S_{1}}{\cd,\cd,\cd}\cup \bn{S_{2}}{\cd,\cd,\cd}\cup \bn{S_{3}}{\cd,\cd,\cd}$\label{line:mm5}
                \For{$(i,j,k)$ in $\{ (1,2,3), (2,3,1), (3,1,2) \}$}
                \For{$(F_1,L_1,R_1) \in \binom{S_{i,j}}{\cdot,\cdot,\cdot}$, $(F_2,L_2,R_2) \in \binom{S_{i,k}}{\cdot,\cdot,\cdot}$}
                
                \State $\cntthree \gets 0$
                \For{$(F_3,L_3,R_3) \in \binom{S_i \cap F}{\cdot,\cdot,\cdot}$}
                \State $\cntthree\gets\cntthree+\forestDP(G[S_i],\omega,F_3,L_3,R_3,k_i,m'_i,W_i)$ 
                \EndFor
                \State Add an edge $e$ between vertices $(F_1,L_1,R_1)$ and $(F_2,L_2,R_2)$ of $H$ 
                \State Assign weight $\cnt3\ (\mathrm{mod}\ 2^{n-k-m'})$ to the edge $e$ \label{line:mm12}
                \EndFor
                \EndFor
                \State $\cntzero \gets$ sum over the product of the three edges of all triangles in $H$ \label{line:mm13}
                \State $\cnt\gets\cnt+\cntzero$
                \EndFor
                \EndFor
                \EndFor
                
                \If{$\cnt \not\equiv 0\ (\mathrm{mod}\ 2^{n-k-m'})$} \label{line:mm17}
                \State \Return a FVS of $G[\{v_1,\lds,v_i\}]$ of size $\le k$ satisfying $\deg(F)\le\bar dk$, constructed by self-reduction
                \EndIf
                \EndFor
                \State \Return \Inf
        \end{algorithmic}
\end{algorithm}

\BL\label{lem:btd-mm}
For any constant $\e>0$, the \BFVS problem with parameters $k$ and $\bar d$ can be solved in time $\Os(3^{(1-\min\{ (2/3)^{\bar d}, (3-\om)(2/3)^{\bar d} + (2\om-3)3^{-\bar d} \}+o(1))k})$.
\EL
\BP
Let $f_1,f_2$ be the values from \Cref{lem:3way}, and let $f_3:=1-3f_1-3f_2$, so that $(f_3-o(1))k\le|S_{1,2,3}| \le (f_3+o(1))k$. For each of the $\Os(3^{f_3+o(1)})$ iterations on Line~\ref{line:mm-for0}, building the graph $H$ (Lines~\ref{line:mm5}~to~\ref{line:mm12}) takes time $\Os(3^{2f_2+f_1+o(k)})$, and running matrix multiplication (Line~\ref{line:mm13}) on a graph with $\Os(3^{f_2+o(k)})$ vertices to compute the sum over the product of the three edges of all triangles takes time $\Os(3^{\om f_2+o(k)})$. Therefore, the total running time is
\begin{align*}
\Os(3^{f_3+o(k)} (3^{2f_2+f_1+o(k)} + 3^{\om f_2+o(k)})) &= \Os(3^{f_3+2f_2+f_1+o(k)} + 3^{f_3+\om f_2+o(k)})
\\&= \Os(3^{1-f_2-2f_1+o(k)} + 3^{1-(3-\om)f_2-3f_1+o(k)})
\\&= \Os(3^{1-(f_2+2f_1)+o(k)} + 3^{1-(3-\om)(f_2+2f_1)-(2\om-3)f_1+o(k)})
\\&\le \Os(3^{1-(2/3)^{\bar d}+o(k)} + 3^{1-(3-\om)(2/3)^{\bar d}-(2\om-3)\cd3^{-\bar d}+o(k)}),
\end{align*}
where the last inequality uses Conditions (1a) and (1b) of \Cref{lem:3way}, and the fact that $2\om-3\ge0$.
\EP

\bibliographystyle{alpha}
\bibliography{ref}

\end{document}